\documentclass[11pt]{article}
\usepackage[letterpaper,margin=1in]{geometry}
\usepackage{amsmath,amsthm,amssymb}
\usepackage{graphicx}
\usepackage{algorithmic}
\usepackage{enumerate}
\usepackage{latexsym}
\usepackage{caption}
\usepackage{multirow}
\usepackage{paralist}
\usepackage[dvipsnames]{xcolor}
\usepackage[normalem]{ulem}
\usepackage{authblk}

\usepackage{booktabs} % For formal tables
\usepackage[noend,ruled,noline]{algorithm2e} % For algorithms

\usepackage{enumitem}
\usepackage{mathtools}
\usepackage{amsmath,amsthm}
\SetAlFnt{\small}
\SetAlCapFnt{\small}
\SetAlCapNameFnt{\small}
\SetAlCapHSkip{0pt}

\IncMargin{-\parindent}

\usepackage{natbib}

 \newtheorem{theorem}{Theorem}
 
  \newtheorem{lemma}[theorem]{Lemma}
  
 \newtheorem{definition}[theorem]{Definition}
  \newtheorem*{definition*}{Definition}

  \theoremstyle{remark}
 
   \usepackage{thm-restate}

\usepackage{tikz}
\usepackage{verbatim}
\usetikzlibrary{shapes,arrows,fit,calc,positioning}
\usetikzlibrary{arrows}
\usetikzlibrary{decorations.markings}

    \def \N {\mathbb{N}}
    \def \R {\mathbb{R}}

\renewcommand{\geq}{\ensuremath{\geqslant}}
\renewcommand{\leq}{\ensuremath{\leqslant}}

\newcommand{\predict}{\hat{o}}

\newcommand{\approxratio}{r}

\def \R {\mathbb{R}}
\def \N {\mathbb{N}}

\usepackage{soul}

\newcommand{\minmaxsingle}{\textsc{MinMaxP }}
\newcommand{\minmaxmech}{\textsc{Minimum Bounding Box }}
\newcommand{\minsummech}{\textsc{Coordinatewise Median with Predictions }}

\newcommand{\cmp}{\text{CMP}}
\newcommand{\OPT}{C^e(o, P)}

\newcommand{\fo}{o}
\newcommand{\fa}{f}
\newcommand{\falg}{f}
\newcommand{\fopt}{o}
\newcommand{\famPR}{\mathcal{P}^R_{\text{coa}}}
\newcommand{\famPC}{\mathcal{P}^C_{\text{coa}}}
\newcommand{\famP}{COA}
\newcommand{\coaR}{\mathcal{P}^R_{\text{ca}}}
\newcommand{\coaC}{\mathcal{P}^C_{\text{ca}}}
\newcommand{\coa}{\textsc{CA}}
\newcommand{\ooaR}{\mathcal{P}^R_{\text{oa}}}
\newcommand{\ooaC}{\mathcal{P}^C_{\text{oa}}}
\newcommand{\ooa}{\textsc{OA}}

\newcommand{\fpred}{\hat{o}}

\DeclarePairedDelimiter{\floor}{\lfloor}{\rfloor}

\newcommand{\thres}{\lceil \frac{(1+c)n}{2} \rceil}

\allowdisplaybreaks

\title{Learning-Augmented Mechanism Design:\\  Leveraging Predictions for Facility Location}
\author[a]{Priyank Agrawal\thanks{pa2608@columbia.edu}}
\author[a]{Eric Balkanski\thanks{eb3224@columbia.edu}}
\author[b]{Vasilis Gkatzelis\thanks{gkatz@drexel.edu}}
\author[a]{Tingting Ou\thanks{ to2372@columbia.edu}}
\author[b]{Xizhi Tan\thanks{xizhi@drexel.edu}}
\affil[a]{Columbia University, IEOR}
\affil[b]{Drexel University, Computer Science}

\begin{document}

\date{}
\maketitle

\begin{abstract}
In this work we introduce an alternative model for the design and analysis of strategyproof mechanisms that is motivated by the recent surge of work in ``learning-augmented algorithms''. Aiming to complement the traditional approach in computer science, which analyzes the performance of algorithms based on worst-case instances, this line of work has focused on the design and analysis of algorithms that are enhanced with machine-learned predictions regarding the optimal solution. The algorithms can use the predictions as a guide to inform their decisions, and the goal is to achieve much stronger performance guarantees when these predictions are accurate (consistency), while also maintaining near-optimal worst-case guarantees, even if these predictions are very inaccurate (robustness). So far, these results have been limited to algorithms, but in this work we argue that another fertile ground for this framework is in mechanism design. 

We initiate the design and analysis of strategyproof mechanisms that are augmented with predictions regarding the private information of the participating agents. To exhibit the important benefits of this approach, we revisit the canonical problem of facility location with strategic agents in the two-dimensional Euclidean space. We study both the egalitarian and utilitarian social cost functions, and we propose new strategyproof mechanisms that leverage predictions to guarantee an optimal trade-off between consistency and robustness guarantees. This provides the designer with a menu of mechanism options to choose from, depending on her confidence regarding the prediction accuracy. Furthermore, we also prove parameterized approximation results as a function of the prediction error, showing that our mechanisms perform well even  when the predictions are not fully accurate.
\end{abstract}

\newpage

\section{Introduction}
For more than half a century, the dominant approach for the mathematical analysis of algorithms in computer science has been worst-case analysis. On the positive side, a worst-case guarantee provides a useful signal regarding the robustness of the algorithm. 
However, it is well-known that the worst-case analysis can be unnecessarily pessimistic, often leading to uninformative bounds or impossibility results that may not reflect the real obstacles that arise in practice. 
These crucial shortcomings of worst-case analysis are making it increasingly less relevant, especially in light of the impressive advances in machine learning that give rise to very effective algorithms, most of which do not admit any non-trivial worst-case guarantees.

Motivated by this tension between worst-case analysis and machine learning algorithms, a surge of recent work is aiming for the best of both worlds by designing robust algorithms that are guided by machine-learned predictions. The goal of this exciting new literature on ``algorithms with predictions'' is to combine the \emph{robustness} of worst-case guarantees with \emph{consistency} guarantees, which prove stronger bounds on the performance of an algorithm whenever the prediction that it is provided with is accurate. 

A lot of this work has focused on dynamic settings, where the input arrives over time and the algorithm needs to make irrevocable decisions before observing the whole input. In contrast to traditional online algorithms, which are assumed to have no information regarding the remaining input, learning-augmented algorithms are provided with a prediction regarding this input. An ideal algorithm is one that performs very well if this prediction is accurate, i.e., it has good consistency, but that also achieves a near-optimal worst-case guarantee, even when the prediction is (arbitrarily) inaccurate, i.e., it has good robustness. A flurry of papers published during the last four years have proposed novel algorithms that achieve non-trivial trade-offs between robustness and consistency (see \citep{mitzenmacher2020algorithms} for a survey of some of the initial results).

In this paper, we argue that another fertile ground for the use of predictions is in \emph{mechanism design}. In contrast to online algorithms, whose information limitations are regarding the future, the main obstacle in mechanism design is the fact that part of the input is private information that only the agents know. To overcome this obstacle, a mechanism can ask the agents to report this information but, since they are strategic, they can misreport it if this leads to an outcome that they prefer. The field of mechanism design has proposed solutions to this problem, but their worst-case guarantees are often underwhelming from a practical perspective. However, if these mechanisms are provided with some predictions regarding (part of) the missing information, this could allow the designer to reach more efficient outcomes despite the incentives of the participants. 

In this paper, we propose a model for designing and evaluating strategyproof mechanisms that are enhanced with predictions, which has the potential to transform the mechanism design literature. At the core of this research agenda lies the following fundamental question:

\begin{center}
     \textit{Can learning-augmented mechanisms achieve good robustness and consistency trade-offs?}
\end{center}

Given a mechanism equipped with a prediction, we can parameterize the worst-case performance guarantee of the mechanism, using the error $\eta$ of the prediction. When the prediction is accurate, i.e., $\eta=0$, then the resulting guarantee is called the \textbf{consistency} of the mechanism. The worst case guarantee irrespective of the error, i.e., the worst-case over all values of $\eta$, is called the \textbf{robustness} of the mechanism.
An ideal mechanism would yield guarantees that gracefully transition from optimal performance when the prediction is correct (perfect consistency) to the best-known worst-case performance as the error increases (perfect robustness), thus capturing the best of both worlds. However, this is impossible in many settings: to achieve perfect consistency a mechanism needs to ``trust'' the prediction, in the sense that it always outputs a solution that is optimal according to the prediction. Yet, if the prediction is incorrect, this solution might be arbitrarily bad, causing unbounded robustness. Our goal is to evaluate how close to this ideal mechanism we can get, i.e., to achieve the best possible trade-off between robustness and consistency.

To exhibit the important benefits of adapting this framework to mechanism design and to gain some insights regarding how predictions could be used by strategyproof mechanisms, we focus on the canonical problem of \emph{facility location}. Apart from the fact that this problem has been the focus of a very long line of literature (e.g., see \citep{PT13,alon2010strategyproof,escoffier2011strategy,feldman2013strategyproof,fotakis2014power,fotakis2016strategyproof,serafino2016heterogeneous,walsh2020strategy} and the recent survey by \cite{chan2021facilitylocation}), it has also been previously used as a natural application domain for exhibiting the potential benefits of new mechanism design models~\citep{PT13,PWZ18}. 

In an instance of the facility location problem in $\R^2$, we are given a set of $n$ agents, with each agent $i$ having a preferred location $p_i \in \R^2$, and we need to choose at which location $f\in \R^2$ to build a facility that will be serving the agents. Once the location of the facility has been determined, each agent suffers a cost that is equal to the Euclidean distance between her preferred location and $f$, and our goal is to choose $f$ that minimizes the social cost. 
In this paper we consider both the minimization of the egalitarian social cost (i.e., the maximum cost over all agents) and the utilitarian social cost (i.e., the average cost over all agents), and the main obstacle is that the preferred location $p_i$ of each agent $i$ is private information to the agent and they can choose to misreport it if this could reduce their cost (e.g., by affecting the facility location choice in their favor). To ensure that the agents will not lie, this research has restricted its attention to strategyproof mechanisms, limiting the extent to which the social cost functions can be optimized. 

\subsection{Our results}

Using the facility location problem as a case study, we exhibit the benefits of adapting the learning-augmented framework in mechanism design. In the facility location problem, the information that the designer is missing is the preferred location of each agent, so our goal is to design practical strategyproof mechanisms that are provided with predictions regarding this information. Rather than assuming that the prediction provides the mechanism with a detailed estimate regarding \emph{all} of the private information, i.e., the preferred location of each agent, we instead consider a less demanding prediction that provides an aggregate signal regarding this information. Specifically, we assume that the mechanisms are provided with only a single point $\hat{o}$, corresponding to a prediction regarding the optimal location for the facility. Note that  this point could readily be computed using the predicted location of each agent, so this prediction requires less information and is easier to estimate. Our results focus on mechanisms that are deterministic and anonymous (they do not discriminate among agents based on their identity), which is a well-studied class of mechanisms in the context of facility location.

\vspace{5pt}
\textbf{Egalitarian social cost.} We first focus on the problem of minimizing the egalitarian social cost and, as a warm-up, on the single-dimensional version of the problem, where all the points lie in $\R$. For this version of the problem, there exists a deterministic and strategyproof mechanism that achieves a $2$-approximation, and this is the best possible approximation in this class of mechanisms. Our result for this case is a deterministic strategyproof mechanism, augmented with a prediction regarding the optimal location for the facility, that achieves the best of both worlds. It returns the optimal solution whenever the prediction is correct  ($1$-consistency), but without sacrificing its worst-case guarantee: it always guarantees a $2$-approximation irrespectively of the prediction quality ($2$-robustness).

We then move on to the two-dimensional version of the problem, for which prior work has produced an optimal deterministic strategyproof mechanism achieving a $2$ approximation \citep{alon2010strategyproof,GH21}. Once again, we are able to achieve perfect consistency, but this time this comes at a small cost in terms of the worst-case guarantee, as we achieve a robustness of $1+\sqrt{2}$. A natural question at this point is whether this loss in robustness was required for us to get the strong consistency guarantee. Our next result shows that this is indeed the case: in fact, to achieve any consistency guarantee better than the trivial $2$ bound, any deterministic anonymous and strategyproof mechanism would have to guarantee robustness no better than $1+\sqrt{2}$. Therefore, our proposed mechanism provides an optimal trade-off between robustness and consistency. Finally, we also provide a more general approximation guarantee for our mechanism as a function of the prediction error, proving that it maintains improved performance guarantees even where the prediction is not fully accurate.

\vspace{5pt}
\textbf{Utilitarian social cost.} We then study the problem of minimizing the utilitarian social cost. The single-dimensional case of this problem can be solved optimally using a deterministic, anonymous, and strategyproof mechanism, so we proceed directly to two-dimensions.
In this case, there is a deterministic, anonymous, and strategyproof mechanism that achieves a $\sqrt{2}$-approximation, which is optimal for this class of mechanisms. We provide a family of mechanisms, parameterized by a ``confidence value'' $c\in [0, 1)$ that the designer can choose depending on how much they trust the prediction. If the designer is confident that the prediction is of high quality, then they can choose a higher value of $c$, which provides stronger consistency guarantees, at the cost of deteriorating robustness guarantees. Specifically, we prove that our deterministic and anonymous mechanism is $\sqrt{2c^2+2}/(1+c)$-consistent and $\sqrt{2c^2+2}/(1-c)$-robust (See Figure~\ref{fig:minsumperformance} for a plot of the trade-off provided by this mechanism). This result exhibits one of the important advantages of the learning-augmented framework, which is to provide the user with more control over the trade-off between worst-case guarantees and more optimistic guarantees when good predictions are available. In fact, we prove that our mechanisms are optimal: no deterministic, anonymous, and strategyproof mechanism can achieve a better trade-off between robustness and consistency guarantees, so our mechanisms exactly capture the Pareto frontier for this problem. Finally, we once again extend our approximation results as a function of the prediction error, verifying that the mechanism achieves improved worst-case guarantees even if the prediction is not fully accurate.

\subsection{Related work}
The learning-augmented mechanism design framework, proposed in this paper, is part of a long literature on alternative performance measures aiming to avoid the limitations of worst-case analysis. A detailed overview of such measures can be found in the ``Beyond the Worst-case Analysis of Algorithms" book edited by  \citet{roughgarden2021beyond}.

\vspace{5pt}
\noindent
\textbf{Learning-augmented algorithms.}
Specifically, this framework extends the recent work on ``learning-augmented algorithms'' (or ``algorithms with predictions''), which focuses on algorithm design and aims to overcome worst-case bounds by assuming that the algorithm is provided with predictions regarding the instance at hand (see  \citep{mitzenmacher2020algorithms} for a survey of the early work in this area). \citet{lykouris2018competitive} introduced   consistency and robustness, which are the two primary metrics used to analyze the performance of algorithms in the learning-augmented design framework. There is a long list of classic algorithmic problems that have been studied in that framework, including online paging \citep{lykouris2018competitive},  scheduling \citep{KPZ18}, and secretary problems \citep{dutting2021secretaries,AGKK20}, optimization problems with covering \citep{BMS20} and knapsack constraints \citep{im2021online}, as well as Nash social welfare maximization \citep{banerjee2020online} and several graph \citep{azar2022online} problems. We note that this line of work has also studied facility location problems \citep{FGGP21,jiang2021online}. However, the crucial difference is that these papers are restricted to non-strategic settings, and the predictions are used to overcome information limitations regarding the future, rather than limitations regarding privately held information.
\cite{medina2017revenue} use bid predictions in auctions to learn reserve prices and yield revenue guarantees as a function of the prediction error, but provide no bounded robustness guarantees.

\vspace{5pt}
\noindent
\textbf{Strategic facility location.}
The facility location problem in the presence of strategic agents has been extensively studied and serves as a canonical mechanism design problem. For example, it was used as the case study that initiated the literature on approximate mechanism design without money \citep{PT13}. For single facility location in one dimension, i.e., on the line, the  mechanism that places the facility at the median over all the preferred points reported by the agents is strategyproof, optimal for the utilitarian social cost objective, and  achieves a $2$-approximation for the egalitarian social cost objective, which is the best approximation achievable by any   deterministic and strategyproof mechanism \citep{PT13}. 
In the two-dimensional Euclidean space, a generalization of this mechanism, the ``coordinatewise median'' mechanism (defined in Section~\ref{sec:prelims}), achieves a $\sqrt{2}$-approximation for the utilitarian objective~\citep{meir2019strategyproof}, and a $2$-approximation for the egalitarian objective~\citep{GH21}. These approximation bounds are both optimal among deterministic and strategyproof mechanisms. 
Additional settings that have been studied include general metric spaces~\citep{alon2010strategyproof} and $d$-dimensional Euclidean spaces \citep{meir2019strategyproof, walsh2020strategy, el2021strategyproofness, GH21}, circles \citep{alon2010strategyproof, meir2019strategyproof}, and trees \citep{alon2010strategyproof, feldman2013strategyproof}. 
Finally, some fundamental results on strategyproof facility location focus on characterizing the space of startegyproof mechanisms. For the single-dimensional case, the characterization of \cite{M80} implies that all deterministic strategyproof mechanisms correspond to the family of ``general median mechanisms'' (defined in Section~\ref{sec:prelims}). For the two-dimensional case, an analogous characterization was subsequently provided by \cite{PSS93}. A more detailed discussion regarding prior work on this problem is provided in the recent survey by \cite{chan2021facilitylocation}.

%%%%%%%%%%%%%%%%%%%%%%%%%%%

\section{Preliminaries}
\label{sec:prelims}
In the single facility location problem in the two-dimensional Euclidean space, the goal is to choose a location $f\in \R^2$ for a facility, aiming to serve a group of $n$ agents. Each agent $i$ has a preferred location $p_i\in \R^2$ and, once the facility location is chosen, that agent suffers a cost $d(f, p_i)$, corresponding to the Euclidean distance between her preferred location and the chosen location. Given a set of preferred locations $P = (p_1, \ldots, p_n)$ for the agents, the two standard social cost functions that prior worked has aimed to minimize are the egalitarian social cost $C^e(f, P) = \max_{p \in P} d(f, p)$ (i.e., the maximum cost over all agents) and the utilitarian social cost $C^u(f, P) = \sum_{p \in P} d(f, p) / n$ (i.e., the average cost over all agents). Given some social cost function, we denote the optimal facility location by $o(P) = (x_o(P), y_o(P))$, or $o$ when $P$ is clear from the context.

In the strategic version of the facility location problem, the preferred location $p_i$ of each agent $i$ is private information. Therefore, to minimize the social cost a mechanism needs to ask the agents to report their preferred locations, $P\in \R^{2n}$, and then use this information to determine the facility location $f(P)$. However, the goal of each agent is to minimize their own cost, so they can choose to misreport their preferred location if that can reduce their cost. A mechanism $f: \R^{2n} \rightarrow \R^2$ is \emph{strategyproof} if truthful reporting is a dominant strategy for every agent, i.e., for all instances $P \in \R^{2n}$, every agent $i \in [n]$, and every deviation $p_i' \in \R^2$, we have that $d(p_i, f(P)) \leq d(p_i, f(P_{-i}, p_i'))$.

A strategyproof mechanism that plays a central role in the strategic facility location problem is the \emph{Coordinatewise Median}  (CM) mechanism. This mechanism  takes as input the locations $P = \{(x_i, y_i)\}_{i \in [n]}$ of the $n$ agents and determines the facility location by considering each of the two coordinates separately. The $x$-coordinate of the facility is chosen to be the median of $\{x_i\}_{i \in [n]}$, i.e., the median of the $x$-coordinates of the agents' locations, and its $y$-coordinate is the median of $\{y_i\}_{i \in [n]}$ (if $n$ is even, we assume the smaller of the two medians is returned).
The more general class of \emph{Generalized Coordinatewise Median}  (\textsc{GCM}) mechanisms take as input the locations $P$ of the $n$ agents, as well as a multiset $P'$ of points that are constant and independent of the locations reported by the agents, and outputs $\textsc{CM}(P \cup P')$. In other words, a GCM mechanism is the coordinatewise median mechanism over the locations of the agents and the additional constant points $P'$ chosen by the designer (often called \emph{phantom} points). Apart from deterministic and strategyproof, any \textsc{GCM} mechanism is also \emph{anonymous}: its outcome does not depend on the identity of the agents, i.e., it is invariant under a permutation of the agents.

In the learning-augmented mechanism design framework proposed in this paper, before requesting the set of preferred locations $P$ from the agents, the designer is provided with a prediction $\predict$ regarding the optimal facility location $o(P)$. The designer can use this information to choose the rules of the mechanism but, as in the standard strategic facility location problem, the final mechanism, denoted $f(P, \hat{o})$, needs to be strategyproof. In essence, if there are multiple strategyproof mechanisms the designer can choose from, the prediction can guide their choice, aiming to achieve improved guarantees if the prediction is accurate (consistency), but retaining some worst-case guarantees (robustness).\footnote{An alternative interpretation is that there is a \emph{single} publicly known mechanism that takes as input the prediction and the agents' reports, and the agents know what the prediction is prior to reporting their preferred locations.}
Consistency and robustness are the standard measures in algorithms with predictions \citep{lykouris2018competitive}. Given some social cost function $C$, a mechanism is \emph{$\alpha$-consistent} if it achieves an $\alpha$-approximation ratio when the prediction is correct ($\hat{o}=o(P)$), i.e.,
\[\max_{P} \left\{\frac{C(f(P, o(P)), P)}{ C(o(P), P)}\right\} \leq \alpha.\] 
A mechanism is \emph{$\beta$-robust} if it achieves a $\beta$-approximation ratio even when the predictions can be arbitrarily wrong, i.e., if \[\max_{P, \hat{o}} \left\{\frac{ C(f(P, \hat{o}), P) }{ C(o(P), P)}\right\} \leq \beta.\]
Note that any known strategyproof mechanism that guarantees a $\gamma$-approximation without predictions directly implies bounds on the achievable robustness and consistency. The designer could just disregard the prediction and use this mechanism to achieve $\gamma$-robustness. However, this mechanism would also be no better than $\gamma$-consistent, since it ignores the prediction. The main challenge is to achieve improved consistency guarantees, without sacrificing too much in terms of robustness.

For an even more refined understanding of the performance of a learning-augmented mechanism, one can also prove worst-case approximation ratios as a function of the prediction error $\eta \geq 0$. In facility location, we let the error $\eta(\hat{o}, P) = d(\hat{o}, o(P))/C(o(P),P)$ be the distance between the predicted optimal location $\hat{o}$ and the true optimal location $o(P)$, normalized by the optimal social cost. Given a bound $\eta$ on the prediction error, a mechanism achieves a $\gamma(\eta)$-approximation if 
\[\max_{P, \hat{o}:~ \eta(\hat{o}, P) \leq \eta}  \left\{\frac{C(f(P, \hat{o}), P)}{ C(o(P), P)}\right\} \leq \gamma(\eta).\]
Note that for $\eta=0$ this bound corresponds to the consistency guarantee and for $\eta\to \infty$ it captures the robustness guarantee. If this bound does not increase too fast as a function of $\eta$, then the mechanism may achieve improved guarantees even if the prediction is not fully accurate.

%%%%%%%%%%%%%%%%%%%%%%%%%%%%

\section{Minimizing the egalitarian social cost}

We start by focusing on the egalitarian social cost function, for which no deterministic and strategyproof mechanism can achieve better than a $2$-approximation, even for the one-dimensional case~\citep{PT13}. As a warm-up, we first provide a deterministic, strategyproof, and anonymous mechanism that is $2$-robust, thus matching the best possible worst-case approximation guarantee, but also $1$-consistent, thus combining the best of both worlds. Then, in Section~\ref{sec:MBB} we extend this mechanism to the two-dimensional case and we prove that it is $1$-consistent and $(1+\sqrt{2})$-robust. 
%This mechanism is thus optimal if the prediction is correct while simultaneously achieving a  $1+\sqrt{2}$ approximation even if the prediction is arbitrarily wrong. 
In Section~\ref{sec:egalitarianlb}, we complement this result by showing that our mechanism is Pareto optimal: we prove that  $1+\sqrt{2}\approx 2.41$ is the best robustness achievable by any deterministic, strategyproof, and anonymous mechanism that achieves any consistency better than $2$ (note that a consistency of $2$ can be trivially achieved by disregarding the predictions and running the coordinatewise median mechanism). Our last result, in Section~\ref{sec:errorEgalit}, goes beyond the robustness and consistency guarantees to provide a more refined bound as a function of the prediction error. Specifically, the approximation achieved by our mechanism degrades linearly from $1$ to $1+\sqrt{2}$ as a function of the prediction error.

\subsection{Warm-up: facility location on the line}

As a warm-up,  we first consider the single-dimensional case of the problem, where $p_i \in \mathbb{R}$ for every agent $i$. 
For this special case, a simple deterministic mechanism that returns the median of the points in $P$ is strategyproof, as well as a $2$-approximation of the egalitarian social welfare, which is the best possible approximation among deterministic and strategyproof mechanisms~\citep{PT13}. Our first result in the learning-augmented framework shows that this worst-case guarantee can be combined with perfect consistency.

Given a prediction $\hat{o}$ regarding the optimal facility location, we propose the \minmaxsingle mechanism, formally defined as Mechanism~\ref{mec:oned}. This mechanism uses the prediction as the default facility location choice, unless the prediction lies ``on the left'' of all the points in $P$ or ``on the right'' of all the points in $P$. In the former case, the facility is placed at the leftmost point in $P$ instead, and in the latter it is placed at the rightmost point in $P$.

\vspace{.2cm}

\begin{algorithm}[H]
\textbf{Input:} points $(p_1, \ldots, p_n) \in \R^n$, prediction $\fpred \in \R$
\uIf{$\fpred \in [\min_{i} p_i, \max_{i} p_i]$}{
 \Return{$\fpred$}
}
\uElseIf{$\fpred < \min_{i} p_i$}{
\Return{$\min_{i} p_i$}}
\Else{
\Return{$\max_{i} p_i$}}
\caption{\minmaxsingle mechanism for egalitarian social cost in one dimension.}
\label{mec:oned}
\end{algorithm}

\vspace{.2cm}

We show that \minmaxsingle  is a deterministic, strategyproof, and anonymous mechanism that is $1$-consistent and $2$-robust.  This mechanism thus achieves the best of both worlds: when the prediction is correct,  it yields an optimal outcome, and when the prediction is incorrect, the approximation factor never exceeds $2$, which is the best-possible worst-case approximation. In essence, the prediction provides a “focal point” that the mechanism can use, allowing it to achieve the optimal consistency without compromising strategyproofness. 

\begin{theorem}\label{thm:1dminmax}
The \minmaxsingle mechanism is deterministic, strategyproof, and anonymous, and it is $1$-consistent and $2$-robust for facility location on the line and the egalitarian objective.
\end{theorem}

\begin{proof}
To show that the mechanism is strategyproof, consider any agent $i$ and, without loss of generality, assume that $p_i\leq \hat{o}$, i.e., that the agent's true preferred location is weakly on the left of the prediction. We consider two cases, depending on whether $p_i$ is weakly greater than all the locations reported by the other agents or not. If it is, this means that if $i$ reported truthfully, the mechanism would place the facility at $p_i$ and $i$ would clearly have no incentive to lie. If, on the other hand, $p_i$ is not weakly greater than all the other reported locations, then the returned location $f$ if $i$ reported the truth would be on the right of $p_i$, i.e., $f>p_i$. However, it is easy to verify that if agent $i$ reported a false point $p'_i< p_i$, this would not affect the outcome, and if she reported a false point $p'_i>p_i$, this could only move $f$ further away from $p_i$. Therefore, it is a dominant strategy for $i$ to report the truth.
An alternative way to verify the fact that this mechanism is strategyproof is by observing that it is actually equivalent to a Generalized Coordinatewise Median (defined in Section~\ref{sec:prelims}) with the set $P'$ of constant points containing $n-1$ copies of the prediction $\hat{o}$. To verify this, note that if $\hat{o}\in [\min_{i} p_i,  \max_i p_i]$, then the median of $P\cup P'$ would be $\hat{o}$, otherwise it would be either the leftmost or the rightmost point of $P$, just like the \minmaxsingle mechanism.

Now, to verify the consistency guarantee, consider any instance where the prediction $\hat{o}$ is accurate. Since the truly optimal location for the egalitarian social welfare is the middle of the leftmost and rightmost point in $P$, then we know that whenever $\hat{o}$ is accurate, it must be that $\fpred \in [\min_{i} p_i, \max_{i} p_i]$. As a result, for any such instance the mechanism will place the facility at the optimal location, $\hat{o}$, leading to a consistency of $1$.

Finally, to verify that this mechanism is $2$-robust, note that the facility location $f$ that the mechanism returns always satisfies $f \in [\min_{i} p_i, \max_{i} p_i]$. As a result, the egalitarian social cost is at most $\max_{i} p_i - \min_i p_i$. On the other hand, the optimal egalitarian social cost is equal to $(\max_{i} p_i - \min_i p_i)/2$, implying the $2$-robustness guarantee.
\end{proof}

\subsection{The minimum bounding box mechanism}\label{sec:MBB}

We now move on to the two-dimensional case, i.e., $p_i \in \R^2$ for every agent $i$, which is the main focus of the paper. 
%The best approximation achievable by any deterministic strategyproof mechanism for the egalitarian objective in two dimensions is $2$, as in the one-dimension special case, and this approximation can be achieved by the Coordinatewise Median (\textsc{CM}) mechanism~\cite{GH21}.
We extend the \minmaxsingle mechanism to this setting by running it separately for each of the two dimensions (see Mechanism~\ref{mec:twod}). An alternative, more geometric, description of this mechanism is that it first computes the minimum axis-parallel bounding box of the set $P$ of agent locations and then places the facility at the location within that box that is closest to the predicted optimal location. We therefore call it the \minmaxmech mechanism.

%If the predicted optimal location lies within that box, then we return this location. Otherwise we project the predicted location to the point.%\xnote{Maybe the algorithm is too detailed now not sure how to simplify it.}
%\vnote{Yes, we can maybe give the 1-dimensional version a name, and just call it here for each of the dimensions}

\vspace{.2cm}

\begin{algorithm}[H]
\textbf{Input:} points $((x_1, y_1), \ldots, (x_n, y_n)) \in \R^{2n}$, prediction $(x_{\fpred}, y_{\fpred})  \in \R^2$

$x_f = \minmaxsingle \hspace{-.13cm} ((x_1, \ldots, x_n) , x_{\fpred} )$\\
$y_f = \minmaxsingle \hspace{-.13cm} ((y_1, \ldots, y_n) , y_{\fpred} )$\\
\Return{$(x_f, y_f)$}
\caption{\minmaxmech mechanism for egalitarian social cost in two dimensions.}
\label{mec:twod}
\end{algorithm}

\vspace{.2cm}

%\vnote{I think we could have a theorem claiming that this is strategyproof, and in the proof itself we could explain why, and also take the opportunity to point out that this can be thought of as an instance of the generalized coordinatewise median mechanism, which would also imply its strategyproofness.}

%\subsection{Performance guarantees}\label{sec:perfEgalit}
We now show that the \minmaxmech mechanism is strategyproof, that it places the facility at the optimal location when the prediction is correct ($1$-consistency), and that it achieves a $1+\sqrt2\approx 2.41$ approximation even when the prediction is arbitrarily incorrect ($1+\sqrt2$-robustness), which is only a slight drop relative to the best achievable approximation, which is $2$.

\begin{theorem}\label{thm:2dminmax}
The \minmaxmech mechanism is deterministic, strategyproof, and anonymous and it is $1$-consistent and $1+\sqrt2$-robust for the egalitarian objective.
\end{theorem}

\begin{proof}
There are two ways to verify the strategyproofness of this mechanism. One intuitive way is to observe that the mechanism treats each dimension separately, running the \minmaxsingle mechanism for each one, so the strategyproofness of that mechanism also implies the strategyproofness of \minmaxmech (since agents want the facility to be as close to their coordinate for each dimension). Alternatively, the strategyproofness can also be verified by the fact that the \minmaxmech mechanism is equivalent to a Generalized Coordinatewise Median mechanism if we let $P'$ contain $n-1$ copies of the prediction $\hat{o}$, as we also observed in the proof of Theorem~\ref{thm:1dminmax}. 

To verify that the mechanism has perfect consistency, we first note that the optimal facility location is always in the convex hull of the points in $P$ (in fact, it is the center of the smallest circle containing all points in $P$, and the radius of this circle corresponds to the egalitarian social cost). This point is clearly within the minimum axis-parallel bounding box (which contains the convex hull), so for any instance where the  prediction $\hat{o}$ is correct, this prediction is in the bounding box, and is thus the location returned by the mechanism, verifying its $1$-consistency.

For the robustness, consider any instance with a set of preferred locations $P$, let $o$ be the optimal facility location and $\OPT = \max_{p_i \in P}d(p_i,o)$. We now consider the circle $c$ with $o$ as its center and the optimal distance $\OPT$ as its radius. Consider the axis-parallel square that has $c$ as its inscribed circle and note that this square contains all the points in $P$ since, by definition of the egalitarian social cost, it must be that all the points in $P$ are contained within the circle $c$, contained in the square. As a result, the minimum axis-parallel bounding box of $P$ is contained in this axis-parallel square. Therefore, since $f$, the location returned by the mechanism, is always within this axis-parallel square  (whose center is $o$ and whose edges are all of length $2 \OPT$) we have $d(o,f) \leq \sqrt{2} \cdot \OPT$, because the points of the square furthest away from its center are its vertices. By the triangle inequality we have that the robustness is at most:
 \[\max_{p_i \in P} d(f, p_i)\leq d(f,o) + \max_{p_i \in P}d(o,p_i) \leq (1+\sqrt{2}) \cdot \OPT.\qedhere\]
\end{proof}

\subsection{Optimality of the mechanism}\label{sec:egalitarianlb}

Since the coordinatewise median \textsc{CM} mechanism  achieves a $2$-approximation for the egalitarian social cost over all instances in two dimensions \citep{GH21}, it is $2$-consistent and $2$-robust. The \minmaxmech mechanism achieves $1$-consistency, but that comes at the cost of the robustness guarantee, which weakens from $2$ to $1+\sqrt{2}\approx 2.41$. A natural question is whether there exists any middle-ground between these two results, i.e, whether some mechanism can combine consistency better than $2$ with robustness better than $1+\sqrt{2}$.

Our next result, Theorem~\ref{thm:LBminmax}, answers this question negatively for deterministic, strategyproof, and anonymous mechanisms.  We  show that any deterministic, strategyproof, and anonymous mechanism that guarantees a consistency better than $2$  must have a robustness no better than $1+\sqrt{2}$, proving the optimality of our mechanism among all the mechanisms that provide consistency guarantees better than $2$.
%\xnote{brought the proof back from appendix}

%\begin{restatable}{theorem}{thmLBminmax}
%\label{thm:LBminmax}
%For any $\epsilon > 0$, there is no deterministic, strategyproof and anonymous mechanism that is $(2-\epsilon)$-consistent and  $(1+\sqrt{2}-\epsilon)$-robust for the egalitarian objective.
%\end{restatable}

\begin{theorem}\label{thm:LBminmax}
There is no deterministic, strategyproof, and anonymous mechanism that is $(2-\epsilon)$-consistent and  $(1+\sqrt{2}-\epsilon)$-robust with respect to the egalitarian objective for any $\epsilon > 0$. 
%\enote{\st{some constant $\epsilon > 0$} any $\epsilon > 0$. (it should be ``any" instead of ``some" and I don't think $\epsilon$ needs to be constant}.
\end{theorem}

\begin{proof}
First, note that any mechanism $f$ with a bounded robustness needs to be \emph{unanimous}, i.e., given a set of points $P$ where all the points are at the same location ($p_i=p_j$ for all $i,j\in [n]$), the mechanism needs to place the facility at that same location, i.e., $f(P)=p_i$. If not, then its cost would be positive, while the optimal cost is zero, by placing the facility at the same location as all the points. Therefore, we can restrict our attention to mechanisms that are unanimous. Using the characterization of \cite{PSS93}, we know that any deterministic, strategyproof, anonymous, and unanimous mechanism in our setting takes the form of a generalized coordinatewise median (GCM) mechanism with $n-1$ constant points in $P'$. Our proof first shows that in order to achieve a consistency better than $2$, this mechanism needs to use the prediction $\hat{o}$ in place of all these $n-1$ constant points. Then, we show that if it does use the prediction $\hat{o}$ in place of all these $n-1$ constant points, its robustness is at least $1+\sqrt{2}$.

For the first part of the proof, consider any GCM mechanism for which the multiset of constant points $P'$ contains at least one point that is not the same as the prediction point, $\hat{o}$. Without loss of generality, assume that this point lies strictly below $\hat{o}$, i.e., that its $y$-coordinate is strictly smaller than $y_{\hat{o}}$ (if this point is strictly on the left, strictly on the right, or strictly above the prediction point, we can directly adjust the argument below to prove the same result).
%that uses the prediction point $\hat{o}$ for $k\leq n-2$ of the $n-1$ constant points in $P'$, and chooses any other points (not on $\hat{o}$) for the remaining $n-1-k$ constant points. Among the constant points that are not on $\hat{o}$, let $c_\ell$ and $c_r$ be the number of them that are strictly on the left of $\hat{o}$ (i.e., their $x$-coordinate is strictly smaller than $x_{\hat{o}}$) and strictly on the right of $\hat{o}$ (i.e., their $x$-coordinate is strictly greater than $x_{\hat{o}}$), respectively. Similarly, let $c_a$ and $c_b$ be the number of them that are strictly above or below. It is easy to verify that $c_\ell+c_r+c_a+c_b\geq n-1-k$, otherwise one of these $n-1-k$ points in $P'$ would have to be located on $\hat{o}$. Without loss of generality, we assume that the largest among these four quantities is $c_b$, implying that $c_b\geq (n-1-k)/4>0$, and since $c_b$ is an integer, we get that $c_b\geq 1$ (i.e., there exists at least one point in $P'$ that is strictly below $\hat{o}$). \enote{Can we simplify the previous argument and just say that if $k\leq n-2$ then there must be at least one point in $P'$ that is strictly above, below, on the left, or on the right of $\hat{o}$?} 
Let $\bar{y}=\max_{p \in P': y_p < y_{\hat{o}}} y_p$ be the maximum $y$-coordinate among the points in $P'$ that are strictly below the prediction, and $\epsilon = y_{\hat{o}} -\bar{y}$
%Let $\epsilon = y_{\hat{o}} - \max_{p \in P': y_p < y_{\hat{o}}} y_p$ be such that the maximum $y$-coordinate among the points in $P'$ that are strictly below the prediction is $y_{\hat{o}}-\epsilon$ 
(there exists at least one point in $P'$ that is strictly below the prediction, so $\epsilon>0$). Then, consider the instance where the set of actual agent points $P$ has $n-1$ points at location $(x_{\hat{o}}, y_{\hat{o}}-\epsilon)$ and 1 point at location $(x_{\hat{o}}, y_{\hat{o}}+\epsilon)$, i.e., $\epsilon$ below the prediction and $\epsilon$ above it, respectively. For this instance, $\hat{o}$ is the correct prediction, as it achieves the optimal egalitarian social cost of $\epsilon$. However, the median of the points in $P\cup P'$ with respect to the $y$-axis is $y_{\hat{o}}-\epsilon$, since there are at least $n$ points in $P\cup P'$ with $y$-coordinate equal to $y_{\hat{o}}-\epsilon$ ($n-1$ points in $P$ and at least one point in $P'$) out of a total of $2n-1$ points in $P\cup P'$. %\st{at most that much, and only $n-1$ points in $P\cup P'$ with a greater $y$-coordinate}}.
Therefore, the egalitarian social cost of the mechanism would be at least $2\epsilon$, since the $y$ coordinate of the facility location would be $y_{\hat{o}}-\epsilon$, but there is an actual agent point on $(x_{\hat{o}}, y_{\hat{o}}+\epsilon)$. Therefore, any such mechanism would have a consistency no better than 2.

Now, we conclude the proof by showing that the robustness of the GCM mechanism that uses the prediction point $\hat{o}$ for all the $n-1$ constant points in $P'$ is no better than $1+\sqrt{2}$. Assume that the prediction $\hat{o}$ is located at $(1,1)$ and consider an instance with $n=3$ points in $P$ located at $(0,1)$, $(1,0)$, and $(-1/\sqrt{2},-1/\sqrt{2})$. In that case, the optimal facility location would be at $(0,0)$ and all the three points in $P$ would have distance 1 from it. However, the set $P'$ contains $n-1=2$ points at $(1,1)$, so the GCM mechanism would place the facility at $(1,1)$, because three of the five points in $P\cup P'$ have $x$-coordinate $1$ and three of the five points in $P\cup P'$ have $y$-coordinate $1$. The distance of this facility location from $(-1/\sqrt{2},-1/\sqrt{2})$ is $1+\sqrt{2}$, which concludes the proof.
\end{proof}

\subsection{Approximation as a function of the prediction error}\label{sec:errorEgalit}

We now extend the consistency and robustness results for \minmaxmech to obtain a refined approximation ratio as a function of the prediction error $\eta$. This result shows that our mechanism achieves improved approximation guarantees not only when $\eta=0$ (which corresponds to the consistency guarantee), but for any value of $\eta$ less than $\sqrt{2}$.
%This result shows that apart from the consistency and robustness guarantee, which correspond to $\eta=0$ and arbitrary $\eta$, respectively, our mechanism even provides non-trivial guarantees even when $\eta>0$ but less than $\sqrt{2}$. 
Specifically, our bound degrades gracefully from the consistency bound of $1$ to the robustness bound of $1 + \sqrt{2}$ as a function of $\eta$.

\begin{theorem}\label{thm:2dminmaxerror}
The \minmaxmech  mechanism achieves a $\min \{ 1 + \eta, 1 + \sqrt{2} \}$ approximation for the egalitarian objective, where $\eta$ is the prediction error.
\end{theorem}
%\begin{proof}
%Given an instance $P$, let $\fpred$ and $o$ be the predicted and actual optimal facility locations. let $\OPT = \max_{p_i \in P}d(0,p_i)$, we have $\eta = \frac{d(\fpred, o)}{\OPT}$. If $\fpred$ is within the minimum bounding box, the output of our mechanism $f$ is just $\fpred$, by triangle inequality we have:
%\[\max_{p_i \in P} d(p_i, f) \leq \max_{p_i \in P} d(p_i, o) + d(o, \fpred) = (1+\eta) \cdot \OPT.\]
%Otherwise if $\fpred$ is outside of the minimum bounding box, by Theorem~\ref{thm:2dminmax} we get that the worst approximation ratio is $\sqrt{2}+1$. Combining the two results we get that the approximation ratio of \minmaxmech, given the prediction error $\eta$, is $\min(1+\eta, 1+ \sqrt{2}).$
%\end{proof}
To obtain a $1 + \eta$ approximation, we aim to bound the distance between the output of the mechanism with the erroneous prediction and the output of the mechanism if it had been given the correct prediction, i.e., the optimal location. We first show a helpful lemma to bound this distance.

\begin{lemma}\label{lem:shiftcm}
Given a set of points $P$ and two predictions $\hat{o}$ and $\tilde{o}$, let $f(P, \hat{o})$ and $f(P, \tilde{o})$ be the respective facility locations chosen by the \minmaxmech mechanism. Then, the distance between these two facility locations is at most the distance between the two predictions, i.e.,
\[d(f(P, \hat{o}), f(P, \tilde{o})) ~\leq~ d(\hat{o}, \tilde{o}).\]
\end{lemma}
\begin{proof}
Let $(x_{\hat{f}}, y_{\hat{f}})=f(P, \hat{o})$ and $(x_{\tilde{f}}, y_{\tilde{f}})=f(P, \tilde{o})$ be facility locations returned by the \minmaxmech mechanism given predictions $\hat{o}$ and $\tilde{o}$, respectively, and let $dx_f=|x_{\hat{f}}-x_{\tilde{f}}|$ and $dy_f=|y_{\hat{f}}-y_{\tilde{f}}|$ be the difference of their $x$ and $y$ coordinates. Similarly, let $dx_o=|x_{\hat{o}}-x_{\tilde{o}}|$ and $dy_o=|y_{\hat{o}}-y_{\tilde{o}}|$ be the corresponding differences for the coordinates of the two predictions. To prove this lemma, we argue that $dx_f\leq dx_o$ and $dy_f \leq dy_o$, implying the desired inequality, since 
\[d(f(P, \hat{o}), f(P, \tilde{o})) ~=~ \sqrt{dx_f^2+dy_f^2}~\leq~\sqrt{dx_o^2+dy_o^2} ~=~ d(\hat{o}, \tilde{o}).\]

We first focus on the $x$-coordinate and, without loss of generality, we assume that $x_{\hat{o}}\leq x_{\tilde{o}}$, i.e., that the first prediction is weakly on the left of the second one. To verify that $dx_f\leq dx_o$, we  proceed with a simple case analysis. If $x_{\hat{o}}\geq \max_i p_i$ or $x_{\tilde{o}}\leq \min_i p_i$, i.e., if the predictions are both on the left of all points in $P$ or both on the right of all points in $P$, then the call to \minmaxsingle mechanism would return the same $x$-coordinate for both cases, i.e., $dx_f=0\leq dx_o$. Otherwise, $x_{\hat{f}}=\max\{\min_i p_i,~ x_{\hat{o}}\}$ and $x_{\tilde{f}}=\min\{\max_i p_i,~ x_{\tilde{o}}\}$. This implies that even in this case $dx_f\leq x_{\tilde{o}}-x_{\hat{o}} = dx_o$. Using the same sequence of arguments for the $y$-coordinate implies that $dy_f\leq dy_o$ and concludes the proof.
\end{proof}

Using this lemma, we are now ready to prove  Theorem~\ref{thm:2dminmaxerror}.

\begin{proof}[Proof of Theorem~\ref{thm:2dminmaxerror}]
Theorem~\ref{thm:2dminmax} already shows that the worst-case approximation of \minmaxmech is at most $(1+\sqrt{2})$, so we just need to prove that it is also at most $1+\eta$.

We first note that the error $\eta$ in the prediction is equal to the normalized distance between the prediction and the actual optimal facility location, i.e., $d(\hat{o},o)/\OPT$, so $d(\hat{o},o)=\eta \cdot \OPT$. Using Lemma~\ref{lem:shiftcm} and substituting $\tilde{o}$ with the actual optimal facility location $o$, i.e., $\tilde{o}=o$, we get 
\begin{equation}\label{ineq:lemmbound}
d(f(P, \hat{o}), f(P, o)) ~\leq~ d(\hat{o}, o) ~=~ \eta \cdot \OPT.
\end{equation}
However, the \minmaxmech mechanism chooses the optimal facility when provided with an accurate prediction (it is $1$-consistent), so $f(P, o)=o$. We can therefore conclude that
\begin{align*}
    C^e(f(P, \fpred), P )  & = \max _{i \in [n]}d(p_i, f(P, \fpred)) \\
    & \leq \max_{i \in [n]} \left( d(p_i, f(P, o)) + d(f(P, o), f(P, \fpred)) \right) \\
    & \leq \max_{i \in [n]} \left( d(p_i, o) + \eta \cdot C^e(o, P)\right) \\
    & \leq C^e(o, P) + \eta \cdot C^e(o, P)\\
    & = (1+\eta) C^e(o, P),
\end{align*}
where the first equation is by definition of the egalitarian social cost, the first inequality uses the triangle inequality, the second inequality uses the fact that $f(P, o)=o$ and Inequality~\eqref{ineq:lemmbound}, and the third inequality uses the definition of the egalitarian social cost.
\end{proof}

%%%%%%%%%%%%%%%%%%%%%%%%%%%%

\section{Minimizing the utilitarian social cost}

In this section, we focus on minimizing the utilitarian social cost function. For the one-dimensional case, returning the median of the preferred points in $P$ is an optimal solution which is also strategyproof. For the two-dimensional case, it is known that the coordinatewise median mechanism guarantees a $\sqrt{2}$-approximation, and no deterministic, anonymous, and strategyproof mechanism can achieve a better guarantee~\citep{meir2019strategyproof}. Our main result in this section is a deterministic, strategyproof, and anonymous mechanism in the learning-augmented framework that uses predictions to achieve an optimal trade-off between robustness and consistency. This mechanism is parameterized by a ``confidence value'' $c\in [0, 1)$ (such that $cn$ is an integer), which is chosen by the designer, depending on how much they trust the prediction. Specifically, we prove that for each choice of $c$, the induced mechanism is $\sqrt{2c^2+2}/{(1+c)}$-consistent and ${\sqrt{2c^2+2}}/{(1-c)}$-robust. If the designer has no confidence in the prediction, setting $c=0$ retrieves the optimal robustness guarantee of $\sqrt{2}$, with a consistency that is also $\sqrt{2}$. For higher values of $c$, the consistency improves beyond $\sqrt{2}$, gradually approaching $1$, at the cost of increased robustness bounds (see Figure~\ref{fig:minsumperformance}). 
In Section~\ref{sec:OptUtil} we show that this trade-off between robustness and consistency provided by our mechanism is, in fact, optimal over all deterministic, strategyproof, and anonymous mechanisms.
Finally, in Section~\ref{sec:errorUtil} we once again provide a more refined bound regarding the approximation that our mechanism achieves as a function of the prediction error.

\subsection{The coordinatewise median with predictions mechanism}

Our \minsummech (\textsc{CMP}) mechanism takes as input the multiset $P$ of the preferred locations reported by the agents, a prediction $\hat{o}$, and a parameter value $c \in [0, 1)$ which captures the designer's confidence in the prediction (such that $cn$ is an integer). The mechanism creates a multiset $P'$ containing $cn$ copies of $\hat{o}$ and outputs $\textsc{CM}(P\cup P')$, i.e., the facility location chosen by the generalized coordinatewise median mechanism whose multiset of constant points $P'$ contains $cn$ copies of the prediction. This set $P'$ provides an interesting way for the designer to introduce a ``bias'' toward the prediction, which increases as a function of the parameter $c$. Specifically, a larger value of $c$ adds more points in $P'$, which can move the median with respect to each coordinate toward the prediction. We use $\falg(P, \fpred, c)$ to denote the facility chosen by \cmp \ with a confidence parameter value of $c$ over preferred points $P$ and prediction $\fpred$. 
Note that for the special case when the confidence parameter is set to $c = (n-1)/n$, i.e., when $P'$ contains exactly $n-1$ copies of the prediction, then \textsc{CMP} reduces to the \minmaxmech mechanism from the previous section.

To prove the robustness and consistency guarantees achieved by this mechanism, we first show that we can, without loss of generality, focus on the class of instances that have the following structure: the outcome of the mechanism is at $(0, 0)$, the optimal outcome is at $(0, 1)$, and every point of $P$ is located either at $(0, 1)$, at $(-x, 0)$, or at $(x, 0)$, for some $x\geq 0$. 

\begin{figure}
    \centering
    \includegraphics[width = 0.48 \textwidth]{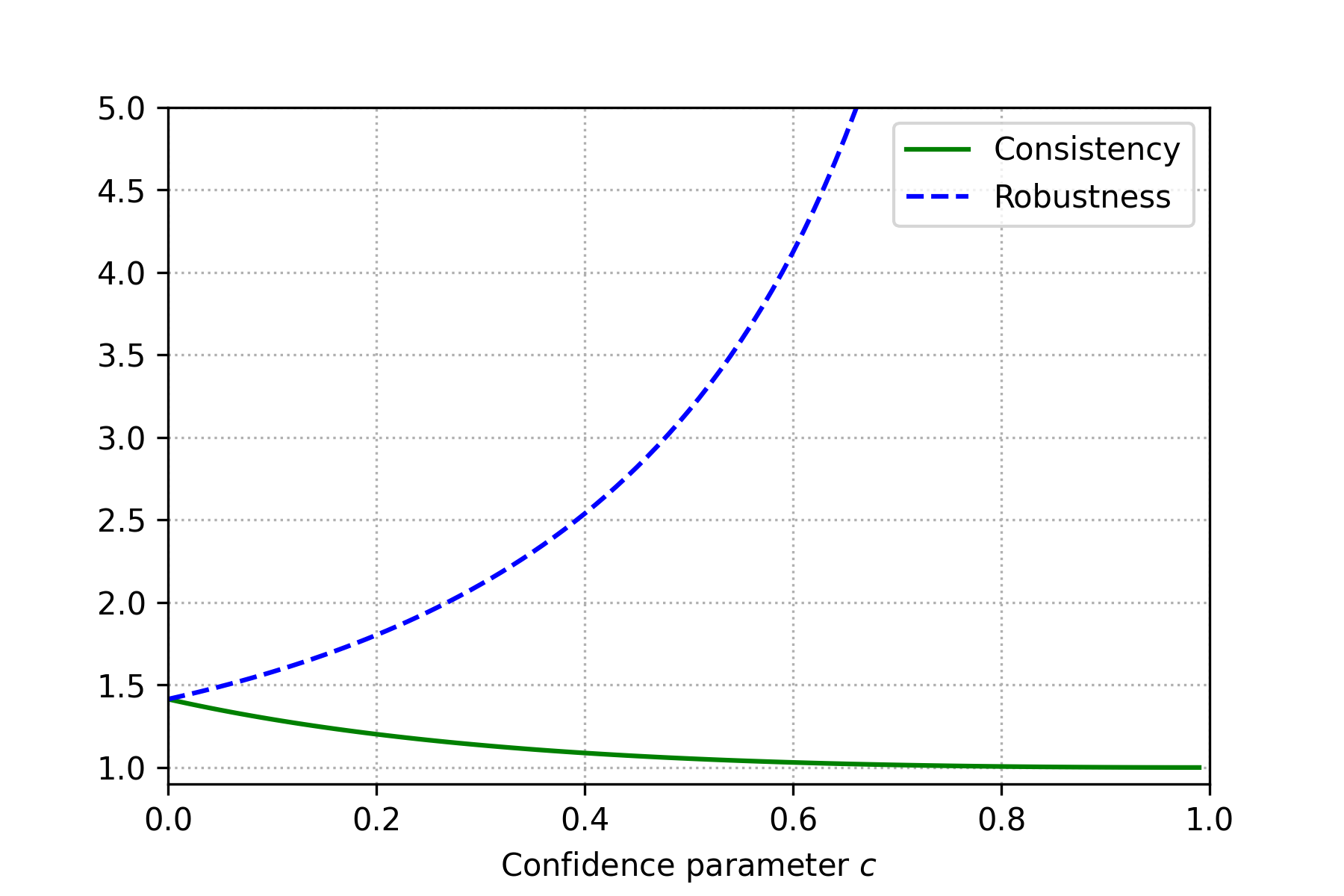}
    \includegraphics[width = 0.48 \textwidth]{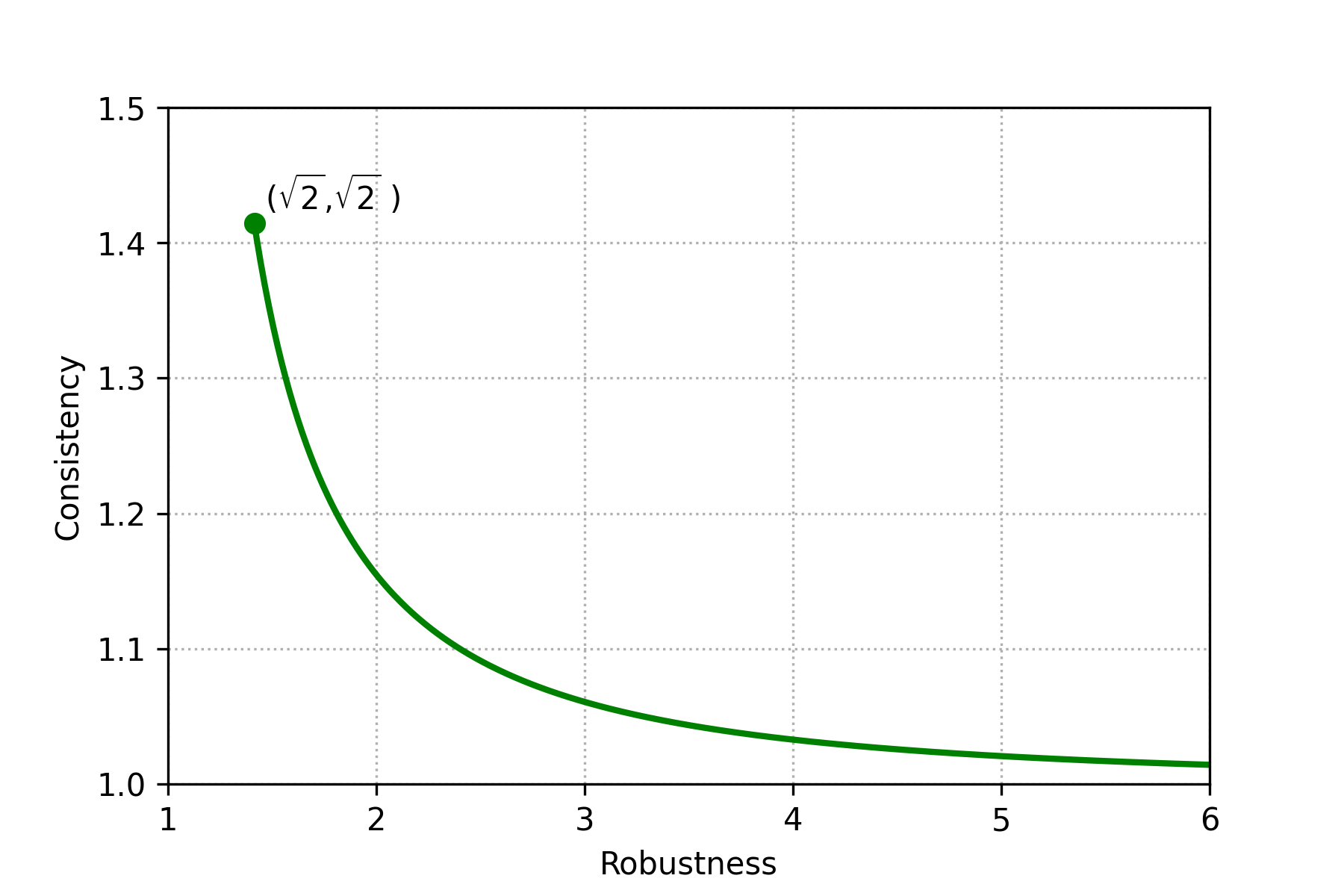}
    \caption{On the left, the consistency and robustness achieved by the \cmp \ mechanism as a function of the confidence parameter $c$. On the right, the optimal trade-off between robustness and consistency, which  is matched by \cmp. Both figures are for the utilitarian social cost objective.}
    \label{fig:minsumperformance}
\end{figure}

\begin{definition}[Clusters-and-OPT-on-Axes Instances]\label{def:cooa}
Given a confidence value $c \in [0, 1)$, consider the class of all instances with predictions $\hat{o}$ and preferred points $P$ (for any number of agents, $n$), such that $\fa(P, \fpred, c) = (0,0)$, $\fopt(P) = (0,1)$, and $p \in \{(0,1), (x,0), (-x,0)\}$ for all $p \in P$ and some $x \in \R_{\geq 0}$. Let $\famPC(c)$ be the subset of these instances where $\fpred = \fopt(P)$ and $\famPR(c)$ be the subset of these instances where $\fpred = (0, 0)$. We refer to these classes of instances as Clusters-and-Opt-on-Axes (\famP) for  consistency and robustness, respectively.
\end{definition}

Our next result is an important technical lemma showing that, if the \cmp\  mechanism with confidence parameter $c$ is $\alpha$-consistent and $\beta$-robust with respect to the classes $\famPC(c)$ and  $\famPR(c)$, respectively, then it is $\alpha$-consistent and $\beta$-robust over all instances. In other words, for any value of $c$, there always exists a worst-case instance within these classes. The structure of our proof resembles an argument used by \cite{GH21} to analyze the worst-case approximation ratio of the standard coordinate-wise median mechanism as a function of $n$ for instances where $n$ is odd (for instances where $n$ is even, a tight bound of $\sqrt{2}$ was already known). However, our argument requires several new ideas to address the fact that the \cmp\ mechanism also depends on the prediction, and to provide bounds not only for robustness, but also for consistency. The resulting argument comprises multiple steps, so we defer the complete proof to Section~\ref{sec:technicallemma}. 
We use $\approxratio(P, \fpred , c)$ to denote the approximation ratio achieved by \cmp \ with parameter $c$ given a multiset of preferred points $P$ and a prediction $\fpred$.

\begin{lemma}\label{lem:mainworstcase}
For any $c \in [0,1)$, the \cmp\ mechanism  with confidence $c$ is $\alpha$-consistent and $\beta$-robust, where $\alpha = \max_{P \in \famPC(c)} \approxratio(P, \fpred = \fopt(P), c)$ and $\beta = \max_{P \in \famPR(c)} \approxratio(P, \fpred = (0,0), c)$.
\end{lemma}

Note that, the $\famPC(c)$ and $\famPR(c)$ classes contain instances with an arbitrary number of agents, yet our robustness and consistency bounds are independent of $n$. We henceforth assume, without loss of generality, that $n$ (the total number of agents) and $cn$ (the number of points in $P'$) are both multiples of $4$ to avoid integrality issues. Indeed, given any parameter value $c$ and any instance where one of these quantities is not a multiple of $4$, we can produce another instance that satisfies both of these conditions and has the same approximation ratio. Specifically, we can achieve this by making four copies of each point in $P$ and $P'$; this would not affect the optimal outcome, nor would it affect the outcome of the mechanism, so the approximation ratio would be the same.

Next, we show that there exists an instance in $\famPC(c)$ such that the \cmp \ mechanism obtains a ${\sqrt{2c^2+2}}/{(1+c)}$ approximation when the prediction is correct and an instance in $\famPR(c)$  where it obtains a ${\sqrt{2c^2+2}}/{(1-c)}$ approximation for some incorrect prediction.

\begin{lemma}\label{lem:lowerbound}
For CMP with confidence $c \in [0,1)$, there exists an instance $P \in \famPC(c)$ such that $r(P, \hat{o} = o(P),c) = \frac{\sqrt{2c^2+2}}{1+c}$ and an instance $Q \in \famPR(c)$ such that $\approxratio(Q, \fpred = (0,0), c) = \frac{\sqrt{2c^2+2}}{1-c}$.
\end{lemma}

\begin{proof}
For the first statement (the consistency bound), consider a multiset of points $P$ that is partitioned into three sets, 
$L$, $R$, and $U$,  
such that $p_i = \left(-\frac{1+c}{1-c},~0\right)$ for $i \in L$ with $|L| = \frac{1+c}{4}n$, $p_i = \left(\frac{1+c}{1-c},~ 0\right)$ for $i \in R$ with $|R| = \frac{1+c}{4}n$, and $p_i = \left(0,1\right)$ for $i \in U$ with $|U| = \frac{1-c}{2}n$.
The optimal location is at $(0,1)$, i.e., $o(P) = (0,1)$, and the optimal cost is $C^u(o(P), P) = \frac{1+c}{2} n \sqrt{1+\left(\frac{1+c}{1-c}\right)^2}$.
Since $f(P, \hat{o}=o(P), c) = (0,0)$, we also have $C^u( f(P,\hat{o}=o(P),c), P) = \frac{1+c}{2} n \cdot \frac{1+c}{1-c} + \frac{1-c}{2} n \cdot 1$. 
Therefore, the consistency, $r(P, \hat{o} = o(P),c)$, of this instance $P$ is:
\[r(P, \hat{o} = o(P), c) = \frac{\frac{1+c}{2} \cdot \frac{1+c}{1-c} + \frac{1-c}{2}}{\frac{1+c}{2} \sqrt{1+\left(\frac{1+c}{1-c}\right)^2}} = \frac{\sqrt{2c^2+2}}{1+c}.\]

For the second statement (the robustness bound), consider the following multiset of points $Q$. Let $L$, $R$, and $U$ be subsets of agents such that $p_i = \left(-\frac{1-c}{1+c},~0\right)$ for $i \in L$ with $|L| = \frac{1-c}{4}n$, $p_i = \left(\frac{1-c}{1+c},~0\right)$ for $i \in R$ with $|R| = \frac{1-c}{4}n$, and $p_i = \left(0,1\right)$ for $i \in U$ with $|U| = \frac{1+c}{2}n$. Note that instances $P$ and $Q$ are very similar, except for the locations of the clusters on the $x$-axis and the number of points on each cluster. Given again $o(Q) = (0,1)$ and $f(Q, \hat{o} = (0,0), c) = (0,0)$, we have $C^u(o(Q), Q) = \frac{1-c}{2} n  \sqrt{1+\left(\frac{1-c}{1+c}\right)^2}$ and $C^u(f(Q, \hat{o}=(0,0), c), Q) = \frac{1-c}{2} n  \cdot \frac{1-c}{1+c} + \frac{1+c}{2}n$, leading to a robustness of
\[r(Q, \hat{o} = (0,0),c) = \frac{\frac{1-c}{2} \cdot \frac{1-c}{1+c} + \frac{1+c}{2}}{\frac{1-c}{2} \sqrt{1+\left(\frac{1-c}{1+c}\right)^2}} = \frac{\sqrt{2c^2+2}}{1-c}. \qedhere\]
\end{proof}

We can now combine Lemma~\ref{lem:mainworstcase} and Lemma~\ref{lem:lowerbound} to obtain the consistency and robustness bounds of the \cmp\ mechanism with respect to the utilitarian objective.

\begin{theorem}\label{thm:minsumconsistencyrobusntess}
The CMP mechanism with parameter $c \in [0,1)$ is $\frac{\sqrt{2c^2+2}}{1+c}$-consistent and $\frac{\sqrt{2c^2+2}}{1-c}$-robust for the utilitarian objective.
\end{theorem}

\begin{proof}
We first argue the consistency guarantee. From Lemma~\ref{lem:mainworstcase} we know that for any confidence value $c\in [0,1)$ and given any instance, we can always find an instance with weakly worse consistency in $\famPC(c)$, i.e., a multiset $P$ such that $o(P) = (0,1)$, $f(P, \hat{o} = o(P),c) = (0,0)$ and there exist $x \in \R_{\geq 0}$ such that $p \in \{(0,1), (x,0), (-x,0)\}$ for all $p \in P$. Note that for any value of $x$, the consistency is maximized when the number of agents on $(0,1)$ is maximized. To see this, note that each agent on $(0,1)$ suffers no cost according to the optimal solution but a cost of $1$ according to the mechanism output, whereas each agent on $(x, 0)$ or $(-x,0)$ suffers a cost of $\sqrt{x^2+1} > x$ according to the optimal solution and a cost of $x$ according to the mechanism output. Since $f(P, \hat{o} = o(P),c) = (0,0)$ and there are $cn$ predicted points on $(0,1)$, the number of agents on $(0,1)$ in the worst case should be $\frac{1-c}{2}n$.
We can then write the consistency as follows:
\[\alpha = \frac{C^u(f(P, \hat{o} = o(P),c), P)}{C^u(o(P),P)} = \frac{\frac{1+c}{2}n \cdot x + \frac{1-c}{2}n}{\frac{1+c}{2}n \cdot \sqrt{1+x^2}} = \frac{1-c + (1+c)x}{(1+c)\sqrt{1+x^2}}.\]
Taking the first derivative with respect to $x$ we get: 
\[\frac{d \alpha}{dx} = \frac{1+c-(1-c)x}{\left(1+c\right)\left(1+x^2\right)\sqrt{1+x^2}}.\]
Solving $\frac{d \alpha}{dx} = 0$ we get that $x =\frac{1+c}{1-c}$. Notice that the denominator of $\frac{d \alpha}{dx}$ is always positive and for any $x< \frac{1+c}{1-c}$, the numerator is positive and for any $x >\frac{1+c}{1-c}$, the numerator is negative, we therefore have that $\alpha$ is maximized at $x = \frac{1+c}{1-c}$. Since the agents on $(x,0)$ and $(-x,0)$ are equidistant from both $o(P)$ and $f(P, \hat{o} = o(P),c)$, the instance is identical to the lower bound instance $P$ in Lemma~\ref{lem:lowerbound}. Therefore we have
\[\alpha = r(P, \hat{o} = o(P), c) = \frac{\sqrt{2c^2+2}}{1+c}.\]

The proof for the robustness guarantee is similar. From Lemma~\ref{lem:mainworstcase} we know that for any confidence value $c\in [0,1)$ and given any instance, we can always find an instance with weakly worse robustness in $\famPR(c)$, i.e., a multiset of points $Q$ such that $o(Q) = (0,1)$, $f(Q, \hat{o} = (0,0),c) = (0,0)$ and there exist $x \in \R_{\geq 0}$ such that $q \in \{(0,1), (x,0), (-x,0)\}$ for all $q \in Q$. Note that,by the exact same reasoning as for consistency, for any value of $x$, the approximation ratio is maximized when the number of agents on $(0,1)$ is maximized. Since again $f(Q, \hat{o} = (0,0),c) = (0,0)$ and there are $cn$ predicted points on $(0,0)$, the number of agents on $(0,1)$ in the worst case should be $\frac{1+c}{2}n$.
We can then write the robustness as follows:
\[\beta = \frac{C^u(f(Q, \hat{o} = (0,0),c), Q)}{C^u(o(Q), Q)} = \frac{\frac{1-c}{2}n \cdot x + \frac{1+c}{2}n}{\frac{1-c}{2}n \cdot \sqrt{1+x^2}} = \frac{1+c + (1-c)x}{(1-c)\sqrt{1+x^2}}.\]
Taking derivative with respect to $x$ and setting it to $0$ we get %\vnote{Why are we setting it to zero already? Same as above}:
\[\frac{d \beta}{dx} = \frac{1-c-(1+c)x}{\left(1-c\right)\left(1+x^2\right)\sqrt{1+x^2}}. \]
Solving $\frac{d \beta}{dx} = 0$ we get that $x =\frac{1-c}{1+c}$. Notice that the denominator of $\frac{d \beta}{dx}$ is always positive and for any $x< \frac{1-c}{1+c}$, the numerator is positive and for any $x >\frac{1-c}{1+c}$, the numerator is negative, we therefore have that $\beta$ is maximized at $x = \frac{1-c}{1+c}$. Since the agents on $(x,0)$ and $(-x,0)$ are equidistant from both $o(Q)$ and $f(Q, \hat{o} = (0,0),c)$, the instance is identical to the lower bound instance $Q$ in Lemma~\ref{lem:lowerbound}. Therefore we have
\[\beta = r(Q, \hat{o} = (0,0),c) = \frac{\sqrt{2c^2+2}}{1-c}.\qedhere\]
\end{proof}

\subsection{Optimality of the mechanism}\label{sec:OptUtil}

The CMP mechanism allows us to achieve consistency better than $\sqrt{2}$, trading it off against robustness. Our next result shows that the trade-off achieved by \cmp \ is optimal.

\begin{theorem}
For any deterministic, strategyproof, and anonymous mechanism that guarantees a consistency of $\frac{\sqrt{2c^2+2}}{1+c}$ for some constant $c\in (0, 1)$, its robustness is no better than $\frac{\sqrt{2c^2+2}}{1-c}$  for  the utilitarian objective.
\end{theorem}

\begin{proof}

We first note that any mechanism $f$ with a bounded robustness needs to be \emph{unanimous}, i.e., given a set of points $P$ where all the points are at the same location ($p_i=p_j$ for all $i,j\in [n]$), the mechanism needs to place the facility at that same location, i.e., $f(P)=p_i$. If not, then its cost would be positive, while the optimal cost is zero, by placing the facility at the same location as all the points. Therefore, we can restrict our attention to mechanisms that are unanimous. Using the characterization of \cite{PSS93}, we know that any deterministic, strategyproof, anonymous, and unanimous mechanism in our setting takes the form of a generalized coordinatewise median  (GCM) mechanism with a set $P'$ of $n-1$ constant points. The rest of our proof first shows that in order to achieve a consistency of $\frac{\sqrt{2c^2+2}}{1+c}$ for some constant $c\in (0, 1)$, the set of $n-1$ points $P'$ used by the GCM mechanism would need to
satisfy the following condition: the number of points in $P'$ that are weakly above $\hat{o}$ (i.e., their $y$-coordinate is at least $y_{\hat{o}}$) need to be at least $cn$ more than the number of points in $P'$ that are strictly below it (i.e., their $y$-coordinate is less than $y_{\hat{o}}$).
Then, we show that if $P'$ satisfies this condition, then the robustness is no better than $\frac{\sqrt{2c^2+2}}{1-c}$.

Consider any GCM mechanism that uses a set  $P'$ of $n-1$ points and let $q_a$ be the number of these points that are weakly above $\hat{o}$ (i.e., their $y$-coordinate is at least $y_{\hat{o}}$) and $q_b$ be the number of points that are strictly below $\hat{o}$ (i.e., their $y$-coordinate is less than $y_{\hat{o}}$). Assume that $q_a - q_b = kn$ where $k<c$ and let $\epsilon>0$ be a constant such that the maximum $y$-coordinate among the $q_b$ points that are strictly below the prediction is $y_{\hat{o}}-\epsilon$. Then, consider the instance where the set of actual agent points $P$ has $(1-k)n/2$ points at $\hat{o}$, and the remaining $(1+k)n/2$ points are divided equally between points $(x_{\hat{o}}-\frac{(1-k)\epsilon}{1+k},~ y_{\hat{o}}-\epsilon)$ and $(x_{\hat{o}}+\frac{(1-k)\epsilon}{1+k},~ y_{\hat{o}}-\epsilon)$\footnote{We assume $kn$ is a multiple of $4$ to avoid integrality issues. If $kn$ is not multiple of $4$, we can modify the instance such that there are $\floor{(1-k)n/2}$ agents at $\hat{o}$ and the remaining agents are divided between the given two points such that each point has at least one agent. It is easy to verify that the optimal facility location  would be at $\hat{o}$ and the mechanism output would have a $y$-coordinate at most $y_{\hat{o}} - \epsilon$. A similar argument also holds for robustness.}. Using the same steps that we used in the proof of Lemma~\ref{lem:lowerbound}, we can verify that the optimal facility location in this case would be at $\hat{o}$ (so the prediction is correct). However, the location where the mechanism places this facility has a $y$-coordinate at most $y_{\hat{o}}-\epsilon$. This is because the number of constant and actual agent points (i.e., points in $P\cup P'$) whose $y$-coordinate is at most  $y_{\hat{o}}-\epsilon$ are $q_b+(1+k)n/2$, while the remaining points are $q_a+(1-k)n/2$. Using the fact that $q_a - q_b = kn$, the median with respect to the $y$-coordinate is at most $y_{\hat{o}}-\epsilon$. This leads to a consistency of $\frac{\sqrt{2k^2+2}}{1+k}$ which is worse than $\frac{\sqrt{2c^2+2}}{1+c}$ (since $k<c$ and the latter is an decreasing function of $c$ on $[0,1)$). 
Therefore, to achieve a consistency of $\frac{\sqrt{2c^2+2}}{1+c}$, the mechanism needs to have $q_a - q_b\geq cn$. 

We now consider any GCM mechanism with $q_a - q_b=kn$ where $k>c$ and show that its robustness is going to be worse than $\frac{\sqrt{2c^2+2}}{1-c}$. To verify this, consider the instance whose set of actual agent points $P$ contains $(1+k)/2$ points on $(x_{\hat{o}},~ y_{\hat{o}}-1)$ and the remaining $(1-k)/2$ points divided equally between $(x_{\hat{o}}-\frac{1-k}{1+k},~ y_{\hat{o}})$ and $(x_{\hat{o}}+\frac{1-k}{1+k},~ y_{\hat{o}})$.  Using the same steps that we used in the proof of Lemma~\ref{lem:lowerbound}, we can verify that the optimal facility location in this case would be at $(x_{\hat{o}},~ y_{\hat{o}}-1)$ (so the prediction is incorrect), but the outcome of the mechanism will have a $y$ coordinate of at least $y_{\hat{o}}$, leading to a robustness of $\frac{\sqrt{2k^2+2}}{1-k}$, which is worse than $\frac{\sqrt{2c^2+2}}{1-c}$ (since $k>c$ and the latter is an increasing function of $c$).

Therefore, the only way to achieve the two desired guarantees is to have $q_a - q_b=cn$ which (running through the same argument and replacing $k$ with $c$) gives you consistency no better than $\frac{\sqrt{2c^2+2}}{1+c}$ and robustness no better than $\frac{\sqrt{2c^2+2}}{1-c}$.
\end{proof}

\subsection{Approximation as a function of the prediction error}\label{sec:errorUtil}

We extend the consistency and robustness results for \cmp \ to obtain an approximation ratio as a function of the prediction error $\eta$. This approximation gracefully degrades from the consistency bound $\frac{\sqrt{2c^2 + 2}}{1+c}$ when $\eta = 0$ to the robustness bound $\frac{\sqrt{2c^2 + 2}}{1-c}$ as a function of $\eta$.

\begin{lemma}\label{lem:shiftcm2}
Given a set of points $P$ and two predictions $\hat{o}$ and $\tilde{o}$, let $f(P, \hat{o})$ and $f(P, \tilde{o})$ be the respective facility locations chosen by the CMP mechanism. Then, the distance between these two facility locations is at most the distance between the two predictions, i.e.,
\[d(f(P, \hat{o}), f(P, \tilde{o})) ~\leq~ d(\hat{o}, \tilde{o}).\]
\end{lemma}
\begin{proof}
Let $(x_{\hat{f}}, y_{\hat{f}}) = f(P,\hat{o})$ and $(x_{\tilde{f}}, y_{\tilde{f}}) = f(P,\tilde{o})$ be facility locations returned by the CMP mechanism given predictions $\hat{o}$ and $\tilde{o}$, respectively,and let $dx_f=|x_{\hat{f}}-x_{\tilde{f}}|$ and $dy_f=|y_{\hat{f}}-y_{\tilde{f}}|$ be the difference of their $x$ and $y$ coordinates. Similarly, let $dx_o=|x_{\hat{o}}-x_{\tilde{o}}|$ and $dy_o=|y_{\hat{o}}-y_{\tilde{o}}|$ be the corresponding differences for the coordinates of the two predictions. To prove this lemma, we argue that $dx_f\leq dx_o$ and $dy_f \leq dy_o$, implying the desired inequality, since 
\[d(f(P, \hat{o}), f(P, \tilde{o})) ~=~ \sqrt{dx_f^2+dy_f^2}~\leq~\sqrt{dx_o^2+dy_o^2} ~=~ d(\hat{o}, \tilde{o}).\]
We first focus on the $x$-coordinate and, without loss of generality, we assume that $x_{\hat{o}}\leq x_{\tilde{o}}$ and $y_{\hat{o}}\leq y_{\tilde{o}}$, i.e., that the first prediction is weakly on the left and bottom of the second one. When the prediction is $\hat{o}$, we have at least half points with $x$-coordinate smaller than or equal to $x_{\hat{f}}$. As we move the prediction to $\tilde{o}$, we move some points by at most $d_{x_o}$, we have at least $k+1$ points with $x$-coordinate smaller than or equal to $x_{\hat{f}} + d_{x_o}$. Thus, we have $d_{x_f} = |d_{\tilde{f}} - d_{\hat{f}}| \leq d_{x_o}$. Using the same sequence of argument for the $y$-coordinate implies that $d_{y_f} \leq d_{y_o}$ and concludes the proof.
\end{proof}

\begin{theorem}\label{lem:ratiopred-minsum}
The
CMP mechanism with parameter $c \in [0, 1)$ achieves a $\min \left\{ \frac{\sqrt{2c^2 + 2}}{1+c} + \eta, \frac{\sqrt{2c^2 + 2}}{1-c} \right\}$-approximation, where $\eta$ is the prediction error, for the utilitarian objective. 
\end{theorem}

\begin{proof}

Theorem~\ref{thm:minsumconsistencyrobusntess} already shows that the worst-case approximation of the CMP mechanism is at most $\frac{\sqrt{2c^2 + 2}}{1-c}$, so we just need to prove that it is also at most $\frac{\sqrt{2c^2 + 2}}{1+c} + \eta$.

We first note that the error $\eta$ in the prediction is equal to the normalized distance between the prediction and the actual optimal facility location, i.e., $d(\hat{o},o)/C^u(o,P)$, so $d(\hat{o},o)=\eta \cdot C^u(o,P)$. Using Lemma~\ref{lem:shiftcm2} and substituting $\tilde{o}$ with the actual optimal facility location $o$, i.e., $\tilde{o}=o$, we get 
\begin{equation}\label{ineq:lemmbound2}
d(f(P, \hat{o},c), f(P, o,c)) ~\leq~ d(\hat{o}, o) ~=~ \eta \cdot C^u(o,P).
\end{equation}
By Theorem~\ref{thm:minsumconsistencyrobusntess}, we also know that the CMP mechanism is $\frac{\sqrt{2c^2 + 2}}{1+c}$ consistent, i.e., $C^u(f(P,o,c),P) \leq \frac{\sqrt{2c^2 + 2}}{1+c} C^u(o,P) $. We can therefore conclude that
\begin{align*}
    C^u(f(P, \fpred,c), P )  & = \frac{1}{n}\sum _{i \in [n]}d(p_i, f(P, \fpred,c)) \\
    & \leq \frac{1}{n}\sum_{i \in [n]} \left( d(p_i, f(P, o,c)) + d(f(P, o,c), f(P, \fpred,c)) \right)\\
     & \leq C^u(f(P,o,c),P) + \eta \cdot C^u(o,P)\\
     & \leq \left(\frac{\sqrt{2c^2 + 2}}{1+c} + \eta \right) \cdot C^u(o,P)
\end{align*}
where the first equation is by definition of the utilitarian social cost, the first inequality uses the triangle inequality, the second inequality uses Inequality~\eqref{ineq:lemmbound2} and the definition of the utilitarian social cost, and the last inequality uses the consistency guarantee of the mechanism, i.e., that $C^u(f(P,o,c),P) \leq \frac{\sqrt{2c^2 + 2}}{1+c} C^u(o,P)$.
\end{proof}

\subsection{Proof of Lemma~\ref{lem:mainworstcase}}\label{sec:technicallemma}

In this section we prove Lemma~\ref{lem:mainworstcase}, which shows that for any confidence parameter $c$ there exists a worst-case set of points $P$ for the performance of \cmp \ within the family of Clusters-and-Opt-on-Axes (\famP) instances, defined in Definition~\ref{def:cooa}. At a high level, we argue that for any multiset of points $P$, there exists a multiset of points $Q$ in \famP \ such that the $\cmp$ mechanism achieves an approximation ratio on $Q$ that is no better than the approximation it achieves on $P$. We construct $Q$ via a series of transformations that starts at an arbitrary $P$ and moves points in a manner that weakly increases the approximation ratio (some of the lemma proofs are deferred to Appendix~\ref{sec:apptechnicallemma}).

This high level approach is similar to the one used by \cite{GH21} \ to obtain a $\sqrt{2} \frac{\sqrt{n^2+1}}{n+1}$ approximation for the coordinate-wise median mechanism in $\R^2$ and the special case where $n$ is odd; the analysis of several of our lemmas is similar to the analysis of this previous result (e.g., Lemmas~\ref{lem:firstytox} and~\ref{lem:doublerotation}), but a crucial difference in our analysis is the impact of the prediction on the mechanism. In particular, as we move points to transform an instance into another instance, this can end up moving the optimal location as well as the outcome of the CMP mechanism in non-trivial ways. To address this issue we introduce multiple new ideas (e.g. Lemmas~\ref{lem:equalsizes}, \ref{lem:ooa},~\ref{lem:geotoaxis}, and \ref{lem:ytogm}).

\begin{figure}[h]
    \centering
    \includegraphics[width= \textwidth]{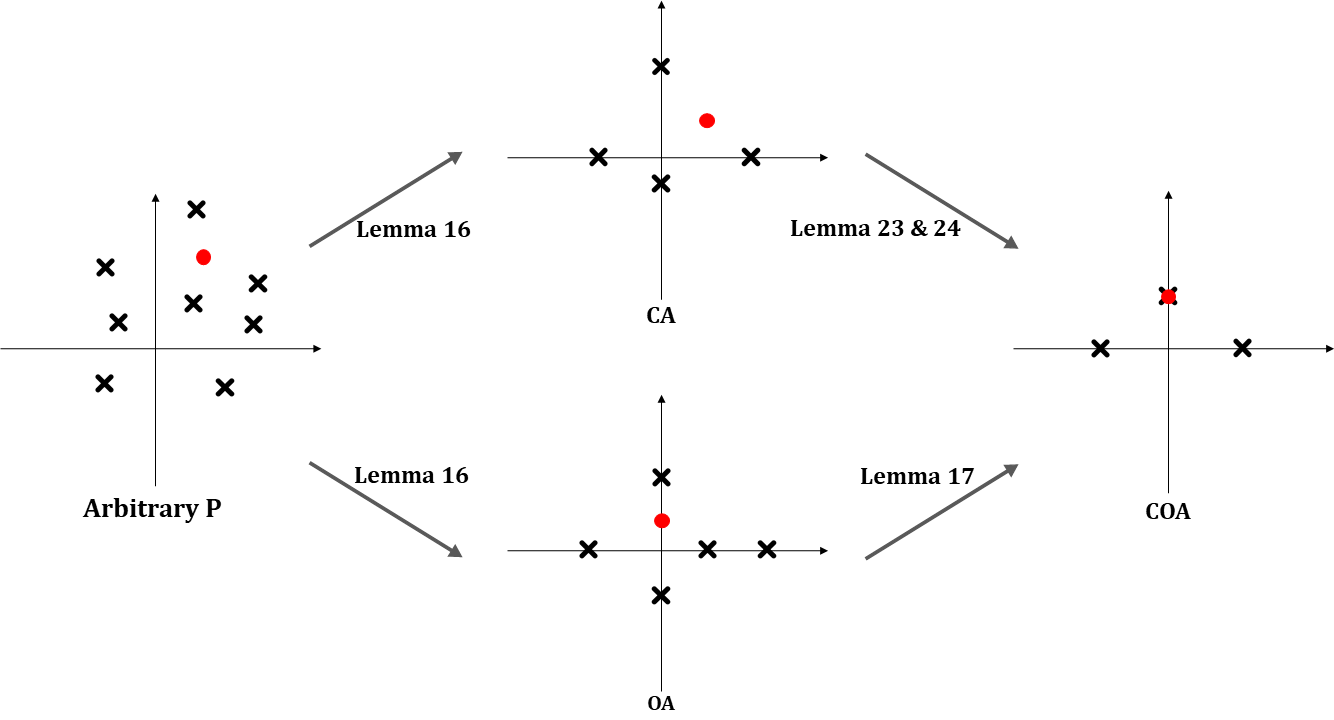}
    \caption{Overview of instance transformations used to prove Lemma~\ref{lem:mainworstcase}.}
    \label{fig:instancetrans}
\end{figure}

We now provide an overview of the series of transformations (see Figure~\ref{fig:instancetrans} for an illustration). In Section~\ref{sec:deffams}, we define the family of \coa \ instances, where the points are all located at four clusters, one on each half-axis, and then the family of \ooa \ instances, where the points and the optimal location are all located on one of the axes. In Section~\ref{sec:coaooa}, we show that an arbitrary instance $P$ can be transformed into either an instance in \coa \ or an instance in  \ooa \ (without improving the approximation ratio). We then show in Section~\ref{sec:ooa} that an instance in \ooa  \ can be transformed to an instance in \famP. The main difficulty is then to transform an instance in \coa \ to an instance in \famP, which we do in Section~\ref{sec:coa}. Finally we combine all these steps to prove Lemma~\ref{lem:mainworstcase} in Section~\ref{sec:mainworstcase}.

Throughout this section, we consider instances that consist of a multiset of points $P$ and a prediction $\fpred$ such that the output of $\cmp$ is at the origin and the optimal location lies weakly in the top right quadrant. To verify that this is without loss of generality, note that given any instance, if we move all the points and the prediction in the same direction and by the same distance, we get an instance where both the CMP mechanism and the optimal facility location have also moved along this same direction and by the same distance. Therefore, the approximation factor is invariant to such changes. As a result, given any instance, we can always generate a new instance such that the output of $\cmp$ is at the origin, without affecting the approximation factor. Similarly, given any instance, the points can be reflected across the horizontal and/or the vertical axes to generate a new ``flipped'' instance such that the output of  \cmp \  lies weakly in the top right quadrant without affecting the approximation factor (e.g., if it lies in the bottom left quadrant originally, we can first reflect across the horizontal axis and then the vertical one). We also assume that the prediction is such that $\fpred = (0,0)$ for the robustness analysis. To verify that this is without loss of generality as well, first note that we have already restricted our attention to instances such that the output of the mechanism is at the origin, and then observe that changing the prediction to also be at the origin does not change the output of the \cmp\ mechanism (if the coordinatewise median of $P\cup P'$ is the origin when $\fpred \neq (0,0)$, it will remain the coordinatewise median if we let $\fpred = (0,0)$). Therefore, this does not affect the outcome of the mechanism and, since it also does not affect the optimal facility location, the robustness   remains the same. 

\subsubsection{The \coa\ and \ooa \  families} 
\label{sec:deffams}
We define the \coa\ and \ooa \  families of points $P$.
Let $A_{+x} = \{(x, 0) : x \geq 0\}$ and $A^{>}_{+x} = \{(x, 0) : x > 0\}$ be the set of all points on the positive and strictly-positive $x$-axis. We also define $A_{-x}$, $A^{<}_{-x}$, $A_{+y}$, $A^{>}_{+y}$, $A_{-y}$, $A^{<}_{-y}$ similarly. We define  \coa \ to be the family of instances of points $P$ that satisfy multiple useful properties, the most important of which are that the points are all located at four clusters, one on each half-axis and that the optimal location is not on an axis. 
We denote these families for the consistency and robustness analysis by $\coaC(c)$ and $\coaR(c)$ respectively. Note that these two families are different since, for the consistency analysis, the prediction is at $\fpred = \fo$, for the robustness analysis, the prediction is at $\fpred = (0,0)$.

\begin{definition}
\label{def:coa} 
Consider, for some confidence $c$, and prediction $\fpred$, the family of multisets of points $P$ s.t.
\begin{enumerate}
    \item Output at origin: $\falg(P, \fpred, c) = (0, 0)$,
    \item Opt in top-right quadrant: $y_o(P) \geq x_o(P) > 0$,
    \item No move towards opt: for all $p_i \in P$ and $\epsilon \in (0, 1]$, $\falg((P_{-i}, p_i + \epsilon(o(P) - p_i)), \hat{o}, c) \neq \falg(P, \hat{o}, c)$, 
    \item there exist $x_1, x_2, y_1, y_2 \geq 0$ such that:
    \begin{enumerate}[label=(\alph*)]
        \item Clusters on axes: for all $p \in P$,  $ p \in \{(-x_1, 0), (x_2, 0), (0, -y_1), (0, y_2), o(P)\}$,
        \item Less points in left: $ |\{p \in P : p \in A^<_{-x}\}| <   |\{p \in P : p \in A^>_{+x} \cup \{ \fopt(P) \} \}|$,
        \item Less points in bottom: $ |\{p \in P : p \in A^<_{-y}\}| <  |\{p \in P : p \in A^>_{+y} \cup \{ \fopt(P) \} \}|$, 
        \item $x$-clusters equidistant from opt: if $(-x_1, 0), (x_2, 0)  \in P$, then $x_o + x_1 = x_2 - x_o$, and
        \item $y$-clusters equidistant from opt: if $(-y_1, 0), (y_2, 0)  \in P$, then $y_o + y_1 = y_2 - y_o$.
    \end{enumerate}
\end{enumerate}
Let $\coaC(c)$ and $\coaR(c)$ be this family when $\fpred = o(P)$ and $\fpred = (0,0)$ respectively. These families are called the Clusters-on-Axes (\coa) families for  consistency and robustness.
\end{definition}

We define  \ooa \ to be the family of multisets of points $P$ such that all the points are on one of the two axes (not necessarily in clusters) and the optimal location is on one of the axes (without loss of generality, the $+y$ half axis). The main difference between the \coa \ and \ooa \ families is the location of the optimal location, either on an axis or not.

\begin{definition}
\label{def:ooa} 
Consider, for some confidence $c$, and prediction $\fpred$, the family of  multisets of points $P$ such that (1) Output at origin: $\falg(P, \fpred, c) = (0, 0)$, (2) Opt on $+y$ axis: $x_o(P) = 0$, $y_o(P) > 0$, and (3) Points on axes: for all $p \in P$, $p \in A_{x} \cup A_{y}$. Let $\ooaC(c)$ and $\ooaR(c)$ be this family when $\fpred = o(P)$ and $\fpred = (0,0)$ respectively. These families are called the Optimal-on-Axes (\ooa) families for  consistency and robustness.
\end{definition}

\subsubsection{The worst-case instance is in  \coa \ or \ooa} 
\label{sec:coaooa}

The main lemma in this section shows that an arbitrary instance of a multiset $P$ can be transformed into either an instance in \coa \ or an instance in \ooa \ without improving the approximation ratio of \cmp \ on that instance (Lemma~\ref{lem:coaooa}). 

We first show that if two points are at different locations on the same half-axis and the optimal location is not on an axis, then  there is an instance $Q$ with a strictly worse approximation. This lemma  is used to obtain the clusters on axes property 4.a. for the \coa \ family.

\begin{restatable}{rLem}{lemclusteringpair}
\label{lem:clusteringpair} 
For any points $P$ and confidence $c \in [0,1)$ s.t. $\fa(P, \fpred(P), c) = (0, 0)$, if there are two non-overlapping points $p_i, p_j \in P$, $p_i \neq p_j$  that are on the same half axis, i.e., $A_{+x}, A_{-x}, A_{+y},$ or $A_{-y}$,   and $x_\fo(P), y_\fo(P) > 0$, then there exists points $Q$ such that $r(Q, \fpred(Q), c) > r(P, \fpred(P), c)$ with  predictions  $\fpred(P) = \fo(P)$ and $\fpred(Q) = \fo(Q)$. This inequality also holds with predictions $\fpred(P) = \fpred(Q) = (0, 0)$.
\end{restatable}

The next lemma  shows that if the points are on the axes and the optimal location, if there are at least as many points with $x$-coordinate that is negative than points with $x$-coordinate that is positive, then there is an instance $Q$ with a strictly worse approximation. The same property holds for the $y$-coordinate.

\begin{restatable}{rLem}{lemequalsizes}
\label{lem:equalsizes}
For any points $P$ and confidence $c \in [0,1)$ such that $\fa(P, \fpred(P), c) = (0, 0)$, $y_o(P) \geq x_o(P) > 0$, and $p \in A_x \cup A_y \cup \{\fopt(P)\}$ for all $p \in P$,  if either $|\{p \in P: p \in A^<_{-x}\}| \geq |\{p \in P: p \in A^>_{+x} \cup \{o(P)\}\}|$ or $|\{p \in P: p \in A^<_{-y}\}| \geq |\{p \in P: p \in A^>_{+y} \cup \{o(P)\}\}|$, then there exists points $Q$ such that $r(Q, \fpred(Q), c) > r(P, \fpred(P), c)$ with  predictions  $\fpred(P) = \fo(P)$ and $\fpred(Q) = \fo(Q)$. This inequality also holds with predictions $\fpred(P) = \fpred(Q) = (0, 0)$.
\end{restatable}

We combine Lemma~\ref{lem:clusteringpair} and Lemma~\ref{lem:equalsizes} to obtain that the instances for which \cmp \ obtains the worst consistency and robustness guarantees are  in the \coa \ and \ooa \ families. 

\begin{restatable}{rLem}{lemcoaooa}
\label{lem:coaooa}
For any $c \in [0,1)$,   let $\alpha = \max_{P \in \ooaC(c) \cup \coaC(c) } \approxratio(P, \fopt(P), c)$ and $\beta = \max_{P \in \ooaR(c) \cup \coaR(c)} \approxratio(P, (0,0), c)$. \cmp \ with confidence $c$ is $\alpha$-consistent and $\beta$-robust.
\end{restatable}

\subsubsection{The worst-case instance  in   \ooa \ is also in \famP} 
\label{sec:ooa}

We show that the worst-case instance in \ooa \ is no worse than the worst-case instance in \famP \ for the consistency and robustness  of \cmp. 

\begin{restatable}{rLem}{lemooa}
\label{lem:ooa}
For any $c \in [0,1)$ and $P \in \ooaC(c)$, there exists $Q$ such that either $ \approxratio(Q, \fopt(Q), c) > \approxratio(P, \fopt(P), c)$ or $Q \in \famPC(c) $ and $\approxratio(Q, \fopt(Q), c) \geq \approxratio(P, \fopt(P), c)$. Similarly, for any $P \in \ooaR(c)$, there exists $Q$ such that either $ \approxratio(Q, (0,0), c) > \approxratio(P, (0,0, c)$ or $Q \in \famPR(c) $ and  $ \approxratio(Q, (0,0), c) \geq \approxratio(P, (0,0), c)$.
\end{restatable}

\begin{proof}
Let $c \in [0,1)$ and consider an instance $P \in \ooaC(c)$ with $n$ points, so $y_o(P) > 0$ and $p \in A_x \cup A_y$ for all $p \in P$. Let $d_x=\sum_{(x_i, 0) \in P \cap A_x} |x_i|/|P \cap A_x|$ be the average distance of the points on the $A_x$ axis from the origin. Consider the  instance $Q = (q_1, \ldots, q_n)$ where the points  $p_i \in P \cap A_x$ are replaced by two clusters, one at $(-d_x,0)$ and one at $(d_x,0)$, each containing $|P \cap A_x|/2$ points $q_i$. For the remaining points $p_j \in P \cap A_y \setminus \{(0,0)\}$, we maintain their positions and  set $q_j = p_j$. 

Since the points are perfectly symmetric with respect to the $y$ axis, we have $x_o(Q) = 0 = x_o(P)$. Since the $y$-coordinate of the points are identical in $P$ and $Q$ and since $x_o(Q) = x_o(P)$, we also have $y_o(Q) = y_o(P)$. Thus, $o(Q) = o(P)$. Let $f(P, \fo(P), c) = (x_f(P, \fo(P), c), y_f(P, \fo(P), c))$, and $f(Q, \fo(Q), c) = (x_f(Q, \fo(Q), c), y_f(Q, \fo(Q), c))$.
Since $x_o(Q) = 0$ and $|\{(x_i, y_i) \in Q: x_i < 0\}| = |\{(x_i, y_i) \in Q: x_i > 0\}|$, $x_f(Q, \fo(Q), c) = 0$, and this also holds with $\fpred(Q) = (0,0)$. Since the $y$-coordinate of the points are identical in $P$ and $Q$ and since $y_o(Q) = y_o(P)$, $y_f(Q, \fo(Q), c) = y_f(P, \fo(P), c)$. Thus, $f(Q, \fo(Q), c) = f(P, \fo(P), c) = (0,0)$, and this also holds with $\fpred(Q) = (0,0)$.

In addition, the social cost of the mechanism does not change, because the average distance of the points on the $x$ axis from the origin remained the same, $d_x$. On the other hand, using the convexity of the distance measure,  the optimal social cost weakly improves, so $ \approxratio(Q, \fo(Q), c) \geq \approxratio(P, \fo(P), c)$, and this is also the case with $\fpred(Q) = (0,0)$. If $y = y_o(Q)$ for all $(0, y) \in Q$, then by scaling $Q$ to $Q'$ so that  $y_o(Q') = 1$, we get $Q' \in \famPC(c)$ such that $ \approxratio(Q', \fo(Q'), c) \geq \approxratio(P, \fo(P), c)$. 

Since $x_o(Q) = 0$, if there exists $(0, y) \in Q$  with $y \neq y_o(Q)$, this point can be moved towards $y_o(Q)$ by an arbitrary small  $\epsilon$ and this would strictly worsens the approximation factor (so there exists $Q'$ such that $ \approxratio(Q', \fopt(Q'), c) > \approxratio(P, \fopt(P), c))$ because it either improves both the social cost of the mechanism and the optimal social cost by $\epsilon$ or it improves the optimal social cost by $\epsilon$ and worsens the social cost of the mechanism by $\epsilon$. 

Thus, we have shown that there exists $Q'$ such that $ \approxratio(Q', \fopt(Q'), c) > \approxratio(P, \fopt(P), c)$ or $Q' \in \famPC(c) $ and  $ \approxratio(Q', \fopt(Q'), c) \geq \approxratio(P, \fopt(P), c)$. The analysis for $P \in \ooaR(c)$ follows identically.
\end{proof}

\subsubsection{The worst-case instance  in   \coa \ is also in \famP} 
\label{sec:coa}

In this section, we show that the worst-case instance in \coa \ is no worse than the worst-case instance in \famP \ for the consistency (Lemma~\ref{lem:coaconsistency}) and robustness (Lemma~\ref{lem:coarobustness}) guarantees of \cmp. 

The next lemma shows that if we have an instance $P$ in the \coa \ family, then we can construct another instance $Q$ in the \coa \ family without points on the $-y$ half axis while weakly increasing the approximation ratio, for both the consistency and robustness guarantees.

\begin{restatable}{rLem}{lemfirstytox}
\label{lem:firstytox}
For any confidence $c \in [0,1)$, and $P \in \coaC(c)$, there exists points $Q$ such that either  $r(Q, \fopt(Q), c) > r(P, \fopt(P), c)$ or  $Q \in \coaC(c)$,  $r(Q, \fopt(Q), c) \geq r(P, \fopt(P), c)$, and $q  \in A_x \cup A_{+y} \cup \{o(Q)\}$ for all $q \in Q$.
Similarly, for any confidence $c \in [0,1)$, and $P \in \coaR(c)$, there exists points $Q$ such that either  $r(Q, (0,0), c) > r(P, (0,0), c)$ or  $Q \in \coaR( c)$,  $r(Q, (0,0), c) \geq r(P, (0,0), c)$, and $q  \in A_x \cup A_{+y} \cup \{o(Q)\}$ for all $q \in Q$.
\end{restatable}

Using the above lemma, given an instance in the \coa \ family, we can remove all the points that are located on the $-y$ half axis. The following lemma shows that for the consistency guarantee, if an instance in the \coa \ family does not contain any points on the $-y$ half axis, then the number of points on the $-x$ axis is larger than or equal to the number of points at the optimal location.

\begin{restatable}{rLem}{lemsizeleft}
\label{lem:sizeleft}
For any $c \in [0,1)$, consider an instance $P \in \coaC(c)$ such that  $p \in A_x \cup A_{+y}\cup \{o(P)\}$ for all $p \in P$. Then, $|\{p \in P : p \in A_{-x}\}| \geq |\{p \in P : p = o(P)\}|$.
\end{restatable}

Consider again an instance $P$ in the \coa \ family without any points on the $-y$ half axis for the consistency guarantee. Lemma~\ref{lem:doublerotation} shows that we can convert $P$ to an instance $Q$ in the \coa \ family with weakly worse approximation ratio and points on either the $x$-axis or the $+y$ half axis.

\begin{restatable}{rLem}{lemdoublerotation}
\label{lem:doublerotation}
For any confidence $c \in [0,1)$ and $n$ points $P \in \coaC(c)$ such that  $p \in A_x \cup A_{+y}\cup \{o(P)\}$ for all $p \in P$, there exists $n$ points $Q$ such that either $r(Q, \fopt(Q), c) > r(P, \fopt(P), c)$ or  $Q \in \coaC(c)$, $r(Q, \fopt(Q), c) \geq r(P, \fopt(P), c)$ and $q \in A_x \cup A_{+y}$ for all $q \in Q$. 
\end{restatable}

The next lemma shows that for an instance $P$ in the \coa \ family with points only on the $x$-axis and the $+y$ half axis, there exists another instance in the \famP \ family with weakly worse approximation ratio than that of $P$ for the consistency guarantee.

\begin{restatable}{rLem}{lemgeotoaxis}
\label{lem:geotoaxis}
For any confidence $c \in [0,1)$ and $n$ points $P \in \coaC(c)$ such that $p \in A_x \cup A_{+y}$ for all $p \in P$, there exists either $Q \in \famPC(c)$ such that $\approxratio(Q, \fopt(Q), c) \geq \approxratio(P, \fopt(P), c)$ or $Q'$ such that $\approxratio(Q', \fopt(Q'), c) > \approxratio(P, \fopt(P), c)$.
\end{restatable}

We now shift our focus back to the robustness guarantee. The next lemma  states that, for the robustness guarantee, if we have an instance $P$ in the \coa \ family where points are located only on the $x$-axis, $+y$ half axis or at the optimal location, then there is another instance $Q$ in the \famP \ family with weakly worse approximation ratio. Note that such an instance $P$ is the result of Lemma~\ref{lem:firstytox}. 

\begin{restatable}{rLem}{lemytogm}
\label{lem:ytogm}
For any confidence $c \in [0,1)$ and $n$ points $P \in \coaR(c)$ such that  $p \in A_x \cup A_{+y}\cup \{o(P)\}$ for all $p \in P$, there exists either $Q \in \famPR(c)$ such that $\approxratio(Q, (0,0), c) \geq \approxratio(P, \fopt(P), c)$ or $Q'$ and prediction $\fpred$ such that $\approxratio(Q', \fpred, c) > \approxratio(P, (0,0), c)$.
\end{restatable}

Finally, the following two lemmas combine the above lemmas and show how to convert an instance $P$ in the \coa \ family to an instance $Q$ in the \famP \ family while weakly increasing the approximation ratio for the consistency and robustness guarantees, respectively.

\begin{restatable}{rLem}{lemcoaconsistency}
\label{lem:coaconsistency}
For any $c \in [0,1)$  and points $P \in \coaC(c)$, there exists points $Q$ such that either $\approxratio(Q, \fopt(Q), c)>  \approxratio(P, \fopt(P), c)$ or  $Q \in \famPC(c)$ and $ \approxratio(Q, \fopt(Q), c) \geq   \approxratio(P, \fopt(P), c)$. 
\end{restatable}

\begin{restatable}{rLem}{lemcoarobustness}
\label{lem:coarobustness}
For any $c \in [0,1)$  and points $P \in \coaR(c)$, there exists points $Q$ and prediction $\fpred$ such that either $ \approxratio(Q, \fpred, c)>  \approxratio(P, (0,0), c)$ or  $Q \in \famPR(c)$ and $ \approxratio(Q, (0,0), c) \geq   \approxratio(P, (0,0), c)$. 
\end{restatable}

\subsubsection{The worst-case instance  is in \famP} We are ready to prove our main lemma for this section.
\label{sec:mainworstcase}

\begin{proof}[Proof of Lemma~\ref{lem:mainworstcase}]
 Let $\alpha' = \max_{P \in \ooaC(c) \cup \coaC(c) } \approxratio(P, \fopt(P), c)$, $\alpha = \max_{P \in \famPC(c) } \approxratio(P, \fopt(P), c)$, $\beta' = \max_{P \in \ooaR(c) \cup \coaR(c)} \approxratio(P, (0,0), c)$, and $\beta = \max_{P \in \famPR(c)} \approxratio(P, (0,0), c)$.  By Lemma~\ref{lem:coaooa}, \cmp \ with confidence $c$ is $\alpha'$-consistent and $\beta'$-robust. Let $P_{a} \in \ooaC(c) \cup \coaC(c) $ and $P_b \in \ooaR(c) \cup \coaR(c)$ be two instances such that $\approxratio(P_a, \fopt(P_a), c) = \alpha'$ and $\approxratio(P_b, (0,0), c) = \beta'$. 
 
 By Lemma~\ref{lem:coaconsistency} and Lemma~\ref{lem:ooa}, there exists $n$ points $Q$ such that either $ \approxratio(Q, \fopt(Q), c)>  \alpha'$ or  $Q \in \famPC(c)$ and $ \approxratio(Q, \fopt(Q), c) \geq   \alpha'$. If $ \approxratio(Q, \fopt(Q), c)>  \alpha'$, this is a contradiction with the  $\alpha'$-consistency of the mechanism. Otherwise, $Q \in \famPC(c)$ and $\alpha \geq  \approxratio(Q, \fopt(Q), c) \geq   \alpha'$ and \cmp \ is $\alpha$-consistent. Similarly, by Lemma~\ref{lem:coarobustness} and Lemma~\ref{lem:ooa}, there exists $n$ points $Q$ and prediction $\fpred$ such that either $ \approxratio(Q, \fpred, c)>  \beta'$ or  $Q \in \famPR(c)$ and $ \approxratio(Q, (0,0), c) \geq   \beta'$. If $ \approxratio(Q, \fpred, c)>  \beta'$, this is a contradiction with the  $\beta'$-robustness of the mechanism. Otherwise, $Q \in \famPR(c)$ and $\beta \geq  \approxratio(Q, (0,0), c) \geq   \beta'$ and \cmp \ is $\beta$-consistent.
\end{proof}

%%%%%%%%%%%%%%%%%%%%%%%%%%%%%

\section{Conclusion and Future Directions}
Our main thesis in this paper is that the learning-augmented design framework, which has motivated a surge of recent work on ``algorithms with predictions'', can have a transformative impact on the design and analysis of mechanisms in multiagent systems. Such mechanisms face crucial information limitations that hinder the designer from reaching desired outcomes: the most obvious among them is that the designer does not know the participating agents' private information, and these agents may choose to strategically misreport it. Therefore, machine-learned predictions have the potential to address these information limitations and help mechanisms achieve improved performance when the predictions are accurate. To support our thesis, we focused on the canonical problem of facility location and proposed new mechanisms that leverage predictions to achieve a trade-off between robustness and consistency. Depending on how confident the designer is regarding the prediction, our mechanisms provide her with a parameterized menu of options that yield Pareto optimal robustness and consistency guarantees.

There is a loose connection between the learning-augmented mechanism design framework and the line of work on Bayesian mechanism design, which assumes that the agents' private values are drawn from a distribution. This is analogous to the average case analysis for algorithms, which assumes that the input is drawn from a distribution, and the crucial difference with the learning-augmented framework is that it provides no robustness guarantees: in Bayesian mechanism design, the performance of a mechanism is evaluated \textit{in expectation} over this randomness and there are no worst-case performance guarantees in general. This is in contrast to our setting, where we seek performance guarantees even if the predictions are arbitrarily inaccurate and also provide approximation guarantees as a function of the prediction error. Another difference comes from the fact that a lot of the  work on Bayesian mechanism design relaxes the notion of incentives and rather than aiming for strategyproofness, which requires that reporting truthfully is a dominant strategy, it instead aims for Bayesian incentive compatibility, which requires that truthful reporting is an optimal strategy in expectation over the randomness, and assuming everyone else also reports truthfully. Finally, learning these distributions requires a large amount of data about a specific setting (e.g., data about past agents' values for the exact same item that is currently being sold in an auction), whereas machine learning can utilize heterogeneous data (e.g., data about past agents' values for similar items that were previously sold) to obtain predictions, like the ones used in the learning-augmented framework.

The impact of the learning-augmented framework on the design of mechanisms is largely unexplored, so there are multiple important open problems along this research direction. For example, one can revisit any mechanism design problem (both with and without monetary payments) for which we know that strategyproofness leads to impossibility results, aiming to better understand how predictions could help us overcome these obstacles, without compromising the incentive guarantees. We therefore anticipate that this framework will give rise to an exciting new literature that studies classic mechanism design problems from a new perspective.

\appendix
\section{Missing Analysis from Section~\ref{sec:technicallemma}}
\label{sec:apptechnicallemma}

Before we present the missing analysis from Section~\ref{sec:technicallemma}, we introduce some helpful lemmas. The following lemma from \cite{GH21} states that moving a point either away or towards the optimal location (without going past it) does not change the optimal location.

\begin{lemma}[\cite{GH21}] 
\label{lem:optunchanged}
For any $n$ points $P$ and $i \in [n]$, if $p_i \neq \fo(P)$ and  $p'_i \in \{\fo(P)  + t(p_i-\fo(P)) | t \in R_{\geq 0}\}$, then $\fo(P_{-i}, p'_i) = \fo(P)$.
\end{lemma}

The next lemma uses Lemma~\ref{lem:optunchanged} to show that moving a point towards the optimal location strictly worsens the approximation ratio if this movement does not cause the mechanism's output to change and if this output is not optimal. This lemma is used to move points at arbitrary locations to one of the axes.

\begin{restatable}{rLem}{lemtowards}
\label{lem:towards}
For any $n$ points $P$, prediction $\fpred$ and confidence $c \in [0,1)$, $i \in [n]$, $p'_i \in (p_i, \fopt(P)]$, if $\falg(P, \fpred, c) = \falg((P_{-i}, p'_i), \fpred, c)$ and $\approxratio(P, \fpred, c) > 1$, then $\approxratio((P_{-i}, p'_i), \fpred, c) > \approxratio(P, \fpred, c)$.
\end{restatable}
\begin{proof}
Let $Q = (P_{-i}, p'_i)$. From $P$ to $Q$, we only move one point towards $\fopt(P)$ (but not going past it) without changing the output of the mechanism. By definition, $\falg(P, \fpred, c) = \falg(Q, \fpred, c)$, and $\fopt(P) = \fopt(Q)$ by Lemma~\ref{lem:optunchanged}.

Since the $\fo, \fa$ locations are not changed, the amount of decrease in the social cost with respect to the optimal location is $d(p_i, p'_i) > 0$. If the social cost of the mechanism's output does not change or increase, then we have $r(Q, \fpred,c ) > r(P, \fpred,c )$. If the social cost of the mechanism's output decreases, then the amount of decrease is $|d(\fa(P, \fpred, c), p'_i) - d(\fa(P, \fpred, c), p_i)|$. By triangle inequality, $d(p_i, p_i') \geq |d(\fa(P, \fpred, c), p'_i) - d(\fa(P, \fpred, c), p_i)| $. This means that the decrease in social cost with respect to the optimal location is larger than or equal to the decrease with respect to the mechanism's output location. If $r(P, \fpred, c) > 1$, then we have $r(Q, \fpred, c) > r(P, \fpred, c)$.
\end{proof}

The next lemma uses the convexity of the distance function to show that if two points move closer to each other, then the sum of their distance to a third point decreases. 

\begin{restatable}{rLem}{lemconvexity}
\label{lem:convexity} 
Consider three distinct points $p_1, p_2,$ and $p_3$, then for any $\epsilon \in (0, 1)$, $d(p_1, p_3) + d(p_2, p_3) \geq d(p_1 +  \epsilon(p_2 - p_1), p_3) + d(p_2 +  \epsilon(p_1 - p_2), p_3)$. Moreover, if $p_1, p_2,$ and $p_3$ are not collinear, $d(p_1, p_3) + d(p_2, p_3) > d(p_1 +  \epsilon(p_2 - p_1), p_3) + d(p_2 +  \epsilon(p_1 - p_2), p_3)$.
\end{restatable}
\begin{proof}
Note that $ d(p_1 +  \epsilon(p_2 - p_1), p_3) = d((1-\epsilon)p_1 +  \epsilon p_2, p_3) \leq (1 - \epsilon)d(p_1, p_3) + \epsilon d(p_2, p_3)$ where the inequality is by the convexity of the distance function and is strict if $p_1, p_2,$ and $p_3$ are not collinear. Similarly, $d(p_2 +  \epsilon(p_1 - p_2), p_3) \leq (1 - \epsilon)d(p_2, p_3) + \epsilon d(p_1, p_3)$ and we conclude that $d(p_1 +  \epsilon(p_2 - p_1), p_3) + d(p_2 +  \epsilon(p_1 - p_2), p_3) \geq d(p_1, p_3) + d(p_2, p_3)$ (with the inequality being strict if $p_1, p_2,$ and $p_3$ are not collinear).
\end{proof}

Now we present the missing analysis of the section.

\lemclusteringpair*
\begin{proof}
Assume $p_i, p_j \in A_{+x}$ are two non-overlapping points on the $+x$-axis, i.e., $p_i = (x_i, 0)$ and $p_j = (x_j, 0)$ with $x_i > x_j > 0$. Let $Q$ be the instance obtained by moving $p_i$ and $p_j$ towards each other by a distance of  $\epsilon (x_i -x_j)$ where $\epsilon$ is sufficiently small so that the optimal location remains strictly in the top-right quadrant, i.e.,  $x_\fo(Q), y_\fo(Q) > 0$. 
Since $p_i$ and $p_j$ are on the same half-axis, they remain on this same half-axis when we move them towards each other.  Since $p_i$ and $p_j$ remain in the same half-axis and the optimal location remains in the same quadrant, we have that the output of the mechanism does not change, i.e., $\fa(P, \fpred(P), c) = \fa(Q, \fpred(Q), c)$ both when the predictions are $\fpred(P) = \fo(P)$ and $\fpred(Q) = \fo(Q)$ and when they are $\fpred(P) = \fpred(Q) = (0, 0)$.

Since  $p_j$ has distance to the origin which increases by $\epsilon$ and $p_i$ has distance to the origin that decreases by $\epsilon$ and since the output of the mechanism does not change, we have that $C^u(\fa(Q, \fpred(Q), c), Q) = C^u(\fa(P, \fpred(P), c), P)$ both when the predictions are $\fpred(P) = \fo(P)$ and $\fpred(Q) = \fo(Q)$ and when they are $\fpred(P) = \fpred(Q) = (0, 0)$.

Since $\fo(P)$ is not on one of the axes, $p_i$, $p_j$, and $\fo(P)$  are not collinear. By Lemma~\ref{lem:convexity}, we get that $d(p_i, \fo(P)) + d(p_j, \fo(P)) > d(p_i +  \epsilon(p_j - p_i), \fo(P)) + d(p_j +  \epsilon(p_i - p_j), \fo(P)).$ This implies that 
$C^u(\fo(P), P) > C^u(\fo(P), Q) \geq C^u(\fo(Q), Q) $. Since $C^u(\fa(Q, \fpred(Q), c), Q) = C^u(\fa(P, \fpred(P), c), P)$ and 
$C^u(\fo(Q), Q) < C^u(\fo(P), P)$, $r(Q, \fpred(Q), c) > r(P, \fpred(P), c)$. The cases where  $p_i$ and $p_j$ are both on one of the three other half axis follow identically by symmetry.
\end{proof}

\lemequalsizes*
\begin{proof}
Assume $P$ is an instance that satisfies the lemma assumptions and is also such that $|\{p \in P: p \in A^<_{-x}\}| \geq |\{p \in P: p \in A^>_{+x} \cup \{o(P)\}\}|$.  Note that the points are either on the axes or on the optimal location, which is in the top-right quadrant. We consider the instance $Q = \{ q_1, q_2, \ldots, q_n\}$ such that if $p_i \in A_y$, then we have $q_i = p_i$ and if $p_i \notin A_y $, we have $q_i  = p_i - (\epsilon, 0)$ for a small enough $\epsilon$ such that the optimal locations remains in top-right quadrant ($y_o(P), x_o(P) > 0$) and such that $\epsilon \leq \min_{i:p_i \notin A_y} |x_i|$.

Since $\fopt(Q)$ stays in the top-right quadrant as $\fopt(P)$ and the points $q_i$ remain on the same half-axes as $p_i$, we have $\falg(Q,  \fpred(Q), c) = \falg(P, \fpred(P), c) = (0,0)$ both when the predictions are $\fpred(P) = \fo(P), \fpred(Q) = \fo(Q)$ and when they are $\fpred(P) = \fpred(Q) = (0, 0)$.

We now consider the location $q^* = \fo(P) - (\epsilon,0)$.  First, note that $\sum _{p_i \notin A_y, p_i \in P} d(p_i, \fopt(P)) = \sum _{q_i \notin A_y, q_i \in Q} d(q_i,q^*)$  because every point that is not on the $y$ axis has been moved to the left by $\epsilon$, and $q^*$ is also obtained from moving $\fo(P)$ to the left by $\epsilon$. Additionally, we have $\sum _{p_i \in A_y, p_i \in P} d(p_i, \fopt(P)) > \sum _{q_i \in A_y, q_i \in Q} d(q_i,q^*)$ because the points on the $y$ axis are not moved while $q^*$ moves closer to the $y$ axis. Therefore, we obtain that $C^u(q^*, Q) < C^u(\fopt(P), P)$, and we have $C^u(\fopt(Q), Q) \leq C^u(q^*, Q) < C^u(\fopt(P), P)$.

Now we look at the social cost with respect to the mechanism's output location, which is the same for $P$ and $Q$. 
Note that there is the assumption that $|\{p \in P: p \in A^<_{-x}\}| \geq |\{p \in P: p \in A^>_{+x} \cup \{o(P)\}\}|$. We let $n^<_{-x} = |\{p \in P: p \in A^<_{-x}\}|$, $n^>_{+x} = |\{p \in P: p \in A^>_{+x} \}|$ and $n_{o} = |\{p \in P: p = o(P)\}|$. Then we have $n^<_{-x} \geq n^>_{+x} + n_o$.
For the points on the $y$ axis, those points do not move at all when we create $Q$ from $P$, so their social costs do not change.
On the left hand side of $y$ axis, we have increased the social cost of the mechanism by $n^<_{-x} \epsilon$. 
On the right hand side of $y$ axis, the social costs of the points are decreased by at most $(n^>_{+x} + n_o)\epsilon$ because the movement of each point to the left by an $\epsilon$ distance can improve the mechanism's cost by at most $\epsilon$. 
Therefore, because we have $n^<_{-x} \geq n^>_{+x} + n_o$,  the social cost with respect to the mechanism's output location does not decrease, so we have $C^u(\fa(Q,  \fpred(Q), c), Q) \geq C^u(\fa(P,  \fpred(P), c), P)$. Combined with $ C^u(\fopt(Q), Q) <  C^u(\fopt(P), P)$, we get that $r(Q, \fpred(Q), c) > r(P, \fpred(P), c)$ both when the predictions are $\fpred(P) = \fo(P), \fpred(Q) = \fo(Q)$ and when they are $\fpred(P) = \fpred(Q) = (0, 0)$. By symmetry, the case where $|\{p \in P: p \in A^<_{-y}\}| \geq |\{p \in P: p \in A^>_{+y} \cup \{o(P)\}\}|$ follows from the same argument.
\end{proof}

\lemcoaooa*
\begin{proof}
Let $P$ and $\fpred$ be an arbitrary instance of a set of $n$ points and a prediction and let $c \in [0,1)$.  We assume without loss of generality that $\fa(P, \fpred, c) = (0,0)$ and $y_o(P) \geq x_o(P) \geq 0$. First, note that since $\fa(P, \fpred, c) = (0,0)$, we also have that $\fa(P, (0,0), c) = (0,0)$. Thus, $\approxratio(P, \fpred, c) = \approxratio(P, (0,0), c)$. If there exists a point $p \in P$ such that $p \notin A_x \cup A_y \cup \{o(P)\}$, then $p$ can be moved  towards $\fo(P)$ without changing the outcome of the mechanism. Thus, by Lemma~\ref{lem:towards}, either $\approxratio(P, \fpred, c) = 1$ or there exists $Q$ such that $\approxratio(Q, \fpred(Q), c) > \approxratio(P, \fpred, c)$  both when the predictions are $\fpred(Q) = \fo(Q), \fpred = \fo(P)$ and when they are $\fpred(Q) = \fpred(P) = (0, 0)$.

We now assume that $p \in A_x \cup A_y \cup \{o(P)\}$ for all $p \in P$.  If $x_o(P) = 0$, then note that $p \in   A_x \cup A_y \cup \{ o(P)\} = A_x \cup A_y$ for all $p \in P^*$. Thus $P \in \ooaC(n, c)$ when $\fpred = o(P)$ and $P \in \ooaR(n, c)$ when $\fpred = (0,0)$. We now assume that $x_o(P) > 0$, so $y_o(P) \geq x_o(P) > 0$. If there is a point  $p$ that can be moved  towards $\fo(P)$ without changing the outcome of the mechanism, then, again, by Lemma~\ref{lem:towards}, either $\approxratio(P, \fpred, c) = 1$ or there exists $Q$ such that $\approxratio(Q, \fpred(Q), c) > \approxratio(P, \fpred, c)$. Next, we consider the five subproperties of property (4) for the \coa \ family.

Assume  that there is no $x_1, x_2, y_1, y_2 \geq 0$ such that $p \in \{(-x_1, 0), (x_2, 0), (0, -y_1), (0, y_2), o(P)\}$ for all $p \in P$. Since $p \in A_x \cup A_y \cup \{o(P)\}$ for all $p \in P$, there are two points $p_i, p_j$ on the same half-axis and by Lemma~\ref{lem:clusteringpair}, there is an instance $Q$ such that $\approxratio(Q, \fpred(Q), c) > \approxratio(P, \fpred, c)$  both when the predictions are $\fpred(Q) = \fo(Q), \fpred = \fo(P)$ and when they are $\fpred(Q) = \fpred(P) = (0, 0)$. Next, assume that $p \in \{(-x_1, 0), (x_2, 0), (0, -y_1), (0, y_2), o(P)\}$ for all $p \in P$. If $|\{p \in P : p \in A^<_{-x}\} | \geq  |\{p \in P : p \in A^>_{+x} \cup \{ \fopt(P) \} \}|$ or $ |\{p \in P : p \in A^<_{-y}\}| \geq  |\{p \in P : p \in A^>_{+y} \cup \{ \fopt(P) \} \}|$, then by Lemma~\ref{lem:equalsizes} there is an instance $Q$ such that $\approxratio(Q, \fpred(Q), c) > \approxratio(P, \fpred, c)$  both when the predictions are $\fpred(Q) = \fo(Q), \fpred = \fo(P)$ and when they are $\fpred(Q) = \fpred(P) = (0, 0)$.

For the  fourth subcondition, assume that there is a pair of points $p_i = (-x_1, 0) \in P$ and $p_j = (x_2, 0) \in P$ such that $x_o + x_1 \neq x_2 - x_o$. Without loss of generality, we assume $x_o + x_1 < x_2 - x_o $.
Then we construct an instance $Q$ by moving this pair of points in $P$ to the left by $\epsilon$. 
That is, we let $q_i=(-x_1 - \epsilon, 0)$ and $q_j = (x_2 - \epsilon, 0)$ for  $\epsilon$ small enough so that the optimal location remains in the top right quadrant, and let $q_k = p_k$ for any $k \neq i,j$. This improves the optimal cost.   Meanwhile, the mechanism's output and social cost  remain unchanged. Therefore, we have found an instance $Q$ such that $\approxratio(Q, \fpred(Q), c) > \approxratio(P, \fpred, c)$  both when the predictions are $\fpred(Q) = \fo(Q), \fpred = \fo(P)$ and when they are $\fpred(Q) = \fpred(P) = (0, 0)$. The fifth subcondition follows identically by symmetry.
    
We conclude that for any instance $P$ and prediction $\fpred = \fo(P)$ such that $P \notin \ooaC(n, c) \cup \coaC(n, c)$, we have found  an instance  $Q$ such that $\approxratio(Q, \fo(Q), c) > \approxratio(P, \fo(P), c)$, i.e., $Q$ has a worse approximation ratio. Thus the worst-case instance $P$ for the consistency of \cmp \ is such that $P \in \ooaC(n, c) \cup \coaC(n, c)$, which implies that $\alpha = \max_{P \in \ooaC(n, c) \cup \coaC(n, c) } \approxratio(P, \fopt(P), c)$. Similarly, for any instance $P$ and prediction $\fpred$ such that $P \notin \ooaR(n, c) \cup \coaR(n, c)$, we have found an instance  $Q$ such that $\approxratio(Q, (0,0) c) > \approxratio(P, (0,0), c) =  \approxratio(P, \fpred, c)$, thus $\beta = \max_{P \in \ooaR(n, c) \cup \coaR(n, c) } \approxratio(P, (0,0), c)$.
\end{proof}

% lemooa removed from appendix

\lemfirstytox*
\begin{proof}
Let $P \in \coaC(n, c)$. First of all, if $p  \in A_x \cup A_{+y} \cup \{o(P)\}$ for all $p \in P$, then let $Q=P$ and we are done. Now we consider the case where there is some $p \in P$ such that $p \notin A_x \cup A_{+y} \cup \{o(P)\}$, which means that there exists $p \in P$ such that $p \in A^<_{-y}$. 
Because $P \in \coaC(n, c)$, there exist $x_1, x_2, y_1, y_2 \geq 0$ such that for all $p\in P$, we have $p \in \{(-x_1, 0), (x_2, 0), (0, -y_1), (0, y_2), o(P)\}$. 
Additionally, we have $ | \{ p \in P: p \in A^<_{-x} \} | < |\{p \in P: p \in A^>_{+x}\cup \{o(P)\}|$ and that $|\{p \in P : p \in A^<_{-y}\}| <   |\{p \in P : p \in A^>_{+y} \cup \{o(P)\}\}|$ from the definition of $\coaC(n, c)$.
We now create another instance $P^*$ from $P$ by moving one point $(0, -y_1)$ to $(-y_1, 0)$ in $P$, and keeping all other points in $P$ the same. Note that the optimal location of $P^*$ can move elsewhere and it may or may not satisfy $y_o(P^*) \geq x_o(P^*) > 0$. But we can still show that $\falg(P^*, o(P^*), c) = \falg(P, o(P), c) = (0,0)$.

Because $\falg(P, o(P), c) = (0,0)$, we know that $ |\{p \in P: p \in A^>_{+x}\cup \{o(P)\} \}| < \thres -cn$. Therefore, $  | \{ p \in P: p \in A^<_{-x} \} | < |\{p \in P: p \in A^>_{+x}\cup \{o(P)\} \}| < \thres -cn$. Even if the the $cn$ points on the predicted location are now to the left of the $y$-axis in $P^*$, we would still have $ |  \{ p^* \in P^*: p^* \in A^<_{-x} \} | + cn < \thres$. Thus, the $x$-coordinate of the mechanism's output location would still be zero on $P^*$. Using a similar argument and the condition of $\coaC(n,c)$ that $\{ p \in P: p \in A^<_{-y} \} | < |\{p \in P: p \in A^>_{+y}\cup \{o(P)\}|$, we can show that the $y$-coordinate of the mechanism's output location is still zero in $P^*$. Thus we conclude that $\falg(P^*, o(P^*), c) = \falg(P, o(P), c) = (0,0)$, and $C^u(\falg(P, o(P), c), P) = C^u(\falg(P, o(P), c), P^*) = C^u(\falg(P^*, o(P^*), c) , P^*)$. 

Now consider the optimal location $\fopt(P)$ of the instance $P$, and note that $y_o(P) \geq x_o(P) > 0$ because $P \in \coaC(n, c)$. We have $d((0, -y_1), \fopt(P)) \geq d((-y_1, 0), \fopt(P))$ because $y_o(P) \geq x_o(P) > 0$. Therefore, we have $C^u(\fopt(P^*), P^*) \leq C^u(\fopt(P), P^*) \leq C^u(\fopt(P), P)$. 
We discuss two cases based on whether the optimal location moves when we create $P^*$ from $P$.

We first discuss the case where $\fopt(P^*) \neq \fopt(P)$. In this case, we would simply have 
$C^u(\fopt(P^*), P^*) < C^u(\fopt(P), P^*) \leq C^u(\fopt(P), P)$ and that $r(P^*, o(P^*), c) > r(P, o(P), c)$. Therefore, once the optimal location moves, we can find an instance with strictly worse approximation ratio.

We now discuss the remaining case where $\fopt(P^*) = \fopt(P)$. Then, from the above we already have that $y_o(P^*) \geq x_o(P^*) > 0$ and that $r(P^*, o(P^*), c) \geq r(P, o(P), c)$. If we further have $y_1 \neq x_1$, then there are points $(-x_1, 0), (-y_1, 0) \in P^*$ and we can apply Lemma~\ref{lem:clusteringpair} on the instance $P^*$ to find an instance $Q$ such that $r(Q, o(Q), c) > r(P^*, o(P^*), c) \geq r(P, o(P), c)$. Otherwise, we have $y_1 = x_1$. In this case, because we are only moving a point from a cluster to another cluster, we have $p^* \in \{(-x_1, 0), (x_2, 0), (0, -y_1), (0, y_2), o(P^*)\}$ for each $p^* \in P^*$.
Now the instance $P^*$ has satisfied the first two properties of the CA family that $\falg(P^*, o(P^*), c) = (0,0)$ and that $y_o(P^*) \geq x_o(P^*) > 0$. 
If we do not satisfy property (3) of the CA family, then by Lemma~\ref{lem:towards} we can construct an instance $Q$ by moving some point $p \in P^*$ towards $\fopt(P)$ without changing the mechanism's output location and achieve $r(Q, o(Q), c) > r(P^*, o(P^*), c)  \geq r(P, o(P), c) $. We now verify the five subproperties of property (4) of the CA family. Since we just move one point from the cluster on the $-y$ half axis to the cluster on the $-x$ half axis, all the five subconditions clearly hold true except for the second one, which we will verify now. 
Because $P \in \coaC(n,c)$, we have that $ | \{ p \in P: p \in A^<_{-x} \} | < |\{p \in P: p \in A^>_{+x}\cup \{o(P)\}|$. Therefore, after moving one point from the $-y$ half axis to the $-x$ half axis, we would clearly have $ | \{ p^* \in P^*: p^* \in A^<_{-x} \} | \leq |\{p^* \in P^*: p^* \in A^>_{+x}\cup \{o(P^*)\}|$. If we end up with an equality, then we can apply Lemma~\ref{lem:equalsizes} and find an instance $Q$ with $r(Q, o(Q), c) > r(P^*, o(P^*), c)  \geq r(P, o(P), c)$. Thus, the second subcondition is also verified. We then conclude that in this case when $\fo(P^*) = \fo(P)$, either $P^* \in \coaC(n,c)$ with $r(P^*, o(P^*), c)  \geq r(P, o(P), c)$, or we can find $Q$ such that $r(Q, o(Q), c) > r(P, o(P), c)$.
Therefore, if the optimal location doesn't move, moving a point from the $-y$ half axis cluster to the $-x$ half axis cluster will create a new instance again in the CA family with weakly worse approximation ratio, or we can simply find an instance $Q$ with strictly worse approximation ratio. While the optimal location doesn't move, we can iteratively move points from the $-y$ axis cluster to the $-x$ axis cluster while weakly increasing the approximation ratio, until we have no points to move on the $-y$ axis, i.e. we end up with points $Q \in \coaC(n, c)$,  $r(Q, \fopt(Q), c) \geq r(P, \fopt(P), c)$, and $q  \in A_x \cup A_{+y} \cup \{o(Q)\}$ for all $q \in Q$.
Therefore we have proved the lemma when $P \in \coaC(n,c)$. If $P \in \coaR(n,c)$, the proof follows almost identically. 
\end{proof}

\lemsizeleft*
\begin{proof}
Let $c \in [0, 1)$ and $P \in \coaC(n,c)$ such that  $p \in A_x \cup A_{+y}\cup \{o(P)\}$ for all $p \in P$.
Let $k = |\{p \in P : p \in A_{-x}\}|$. Since $P \in \coaC(n,c)$, moving any $p \in P$ towards $o(P)$ would move $\fa(P, \fo(P), c)$. Thus, $|\{p \in P : p \in A_{-x} \cup A_{+x}\}| = |\{p \in P : p \in A_{-x} \cup A_{+y}\}| = \lceil(1+c)n/2 \rceil$ and we get that $|\{p \in P : p \in  A_{+x}\}| = |\{p \in P :  \cup A_{+y}\}| = \lceil(1+c)n/2\rceil - k$ and $|\{p \in P: p = o(P)\}| = n - 2 (\lceil(1+c)n/2\rceil - k) - k = n - 2\lceil(1+c)n/2\rceil + k$. Finally, note that $n - 2\lceil(1+c)n/2\rceil + k \leq n - 2(n/2) + k = k = |\{p \in P : p \in A_{-x}\}|$ and we get the desired result.
\end{proof}

\lemdoublerotation*
\begin{proof}
The main idea is that we can move one point from $\fopt(P)$ to the $+y$ axis, while also moving one point from the $-x$ axis to the $+x$ axis. The new instance will either have a strictly worse approximation ratio, or the ratio is weakly worse but the instance remains in the \coa \ family, so we can apply this paired movement iteratively until there are no points on the optimal location (which means all points are on the axes). 
Also, note that by Lemma~\ref{lem:sizeleft}, we have $|\{p \in P : p \in A_{-x}\}| \geq |\{p \in P : p = o(P)\}|$, so there are enough points on $-x$ axis for us to perform the paired movement and remove all points on the optimal location.

We now formalize the above argument. Let $P = \{ p_1, \ldots, p_n \} \in \coaC(n,c)$ such that  $p \in A_x \cup A_{+y}\cup \{o(P)\}$ for all $p \in P$. 
By Lemma~\ref{lem:sizeleft} we have $|\{ p \in P: p \in A_{-x} \}| \geq |\{p \in P : p = o(P)\}|$. To simplify notation, we let $x_o = x_o(P)$ and $y_o = y_o(P)$. We now define another instance $Q = \{ q_1, \ldots, q_n\}$ to be the instance such that for some $x \geq 0$ and $i_1, i_2 \in N$, we have 
\begin{itemize}
    \item $p_{i_1} = (-x, 0)$ and $ q_{i_1} = (x + 2x_o, 0)$,
    \item $p_{i_2} = \fopt(P) = (x_o, y_o)$ and $q_{i_2} = (0, \sqrt{x_o^2 + y_o^2})$, and 
    \item for $i \neq i_1, i_2$, $p_{i} = q_{i}$.
\end{itemize}

We will now argue that either $Q \in \coaC$ with $r(Q, \fopt(Q), c) \geq r(P, \fopt(P), c)$ and $q \in A_x \cup A_{+y}$ for all $q \in Q$, or we can find an instance with strictly worse approximation ratio than that of $P$.

We first claim that $\falg(Q, o(Q), c) = (0,0)$. Since $P \in \coaC(n,c)$, we have $\falg(P, o(P), c) = (0,0)$, $y_o(P) \geq x_o(P) > 0$, and that there exist $x_1, x_2, y_1, y_2 \geq 0$ such that for all $p \in P$, we have $p \in \{(-x_1, 0), (x_2, 0), (0, -y_1), (0, y_2), o(P)\}$. Additionally, we have $ |\{p \in P : p \in A^<_{-x}\}|  <  |\{p \in P : p \in A^>_{+x} \cup \{ \fopt(P) \} \}|$ and 
$ |\{p \in P : p \in A^<_{-y}\}| <   |\{p \in P : p \in A^>_{+y} \cup \{ \fopt(P) \} \}|$. 
This means that $ |\{p = (x_p, y_p)\in P : x_p < 0\}|  <   |\{p = (x_p, y_p) \in P : x_p > 0 \}| < \thres - cn$. Otherwise, after we count the $cn$ points at the predicted location $o(P)$, the output location of the mechanism $\falg(P, o(P), c)$ will have a positive $x$-coordinate. Note that the instances $P$ and $Q$ only differ at $i_1, i_2$. If $\fopt(Q)$ stays in the top-right quadrant, then it's clear that $\falg(Q, o(Q), c) = (0,0)$. 
However, even if $\fopt(Q)$ moves to the left of the $y$-axis, we have $ |\{q = (x_q, y_q)\in Q : x_q < 0\}| + cn < \thres$ and clearly also that $ |\{q = (x_q, y_q)\in Q : x_q > 0\}|  < \thres$. Thus, the algorithm's output $\falg(Q)$ must have its $x$-coordinate equal to zero. By a similar argument using the fact that $ |\{p \in P : p \in A^<_{-y}\}| <  |\{p \in P : p \in A^>_{+y} \cup \{ \fopt(P) \} \}| < \thres - cn$, we can prove that $\falg(Q, o(Q), c)$ has a $y$-coordinate equal to zero as well. Therefore we conclude that $\falg(Q, o(Q), c)= \falg(P, o(P), c) = (0,0)$, even if the location of $\fopt(Q)$ can be different from that of $\fopt(P)$. 

Next, we prove that $r(Q, o(Q), c) \geq r(P, o(P), c)$. Let $A = \sum _{i \neq i_1} d(p_i, \falg(P, o(P), c))$ and $B = \sum _{i=1}^n d(p_i, \fopt(P)) =  \sum _{i \neq i_2 } d(p_i, \fopt(P))$. The approximation ratio of \cmp \ with a correct prediction $(\fpred = \fopt(P))$ is at least as good as the approximation of the coordinatewise median mechanism (without predictions), so at most $\sqrt{2}$. Thus,
    $A + d(p_{i_1}, \falg(P, o(P), c)) \leq \sqrt{2}B$.
Therefore, 
\begin{align}\label{ineq:dr1}
    [A + d(p_{i_1}, \falg(P, o(P), c))]d(q_{i_2}, \fopt(P)) \leq \sqrt{2}B  d(q_{i_2}, \fopt(P)).
\end{align}

But we have $y_o \geq x_o$, so $d(q_{i_2}, \fopt(P)) \leq \sqrt{2} x_o$. Therefore, from inequality (\ref{ineq:dr1}) we have 
\begin{align}\label{ineq:dr2}
    [A + d(p_{i_1}, \falg(P, o(P), c))]d(q_{i_2}, \fopt(P)) \leq 2B x_o = B(d(q_{i_1}, \falg(P, o(P), c)) - d(p_{i_1}, \falg(P, o(P), c))).
\end{align}

Using inequality (\ref{ineq:dr2}), we have
\begin{align}\label{ineq:dr3}
    &[A + d(p_{i_1}, \falg(P, o(P), c))][B + d(q_{i_2}, \fopt(P)) ] \nonumber\\
  = \ &  AB + Bd(p_{i_1}, \falg(P, o(P), c)) + [A + d(p_{i_1}, \falg(P, o(P), c)) ] d(q_{i_2}, \fopt(P)) \nonumber\\
  \leq \ & AB + Bd(p_{i_1}, \falg(P, o(P), c)) + B(d(q_{i_1}, \falg(P, o(P), c)) - d(p_{i_1}, \falg(P, o(P), c))) \nonumber \\
  = \ & AB + Bd(q_{i_1}, \falg(P, o(P), c)) \nonumber \\
  = \ & (A + d(q_{i_1}, \falg(P, o(P), c))) B.
\end{align}

Note that $A = \sum _{i \neq i_1} d(p_i, \falg(P, o(P), c)) = \sum _{i \neq i_1} d(q_i, \falg(Q, o(Q), c))$ because $\falg(P, o(P), c) = \falg(Q, o(Q), c)$ and $d(p_{i_2}, \falg(Q, o(Q), c)) = d(q_{i_2}, \falg(Q, o(Q), c))$. Additionally, $B = \sum _{i \neq i_2 }{d(p_i, \fopt(P))} = \sum _{i \neq i_1, i_2}{d(q_i, \fopt(P)) + d(p_{i_1}, \fopt(P))} = \sum _{i \neq i_1, i_2} d(q_i, \fopt(P)) + d(q_{i_1}, \fopt(P)) =  \sum _{i \neq i_2} d(q_i, \fopt(P))$.

Thus, using inequality (\ref{ineq:dr3}), we have 
\begin{align}\label{ineq:drfinal}
    r(P, o(P),c) & = \frac{A + d(p_{i_1}, \falg(P, o(P),c))}{B} \nonumber\\
    & \leq \frac{A + d( q_{i_1},\falg(P, o(P),c))}{B + d(q_{i_2}, \fopt(P))}  \hspace{0.5cm} \text{by (\ref{ineq:dr3})} \nonumber \\
    & = \frac{ \sum _{i \neq i_1} d(q_i, \falg(Q, o(Q), c)) + d( q_{i_1},\falg(Q, o(Q), c))}{\sum _{i \neq i_2} d(q_i, \fopt(P)) + d(q_{i_2}, \fopt(P))} \nonumber \\
    & = \frac{C^u(\falg(Q), Q)}{C^u( \fopt(P), Q)} \nonumber \\
    & \leq \frac{C^u( \falg(Q), Q  )}{C^u( \fopt(Q), Q )}  \\
    & = r(Q, o(Q),c) .
\end{align}

The inequality (\ref{ineq:drfinal}) is strict unless $\fopt(P) = \fopt(Q)$. Therefore, if $\fopt(P) \neq \fopt(Q)$ we immediately have $r(P, o(P),c) < r(Q, o(Q),c)$, and we are done. 
Otherwise, we have $\fopt(P) = \fopt(Q)$, so the instance $Q$ now already satisfies the first two properties of the \coa \ family that $\falg(Q, o(Q), c) = (0,0)$ and $\fopt(Q) = \fopt(P) = (x_o, y_o)$ with $y_o \geq x_o > 0$. If property (3) of the \coa \ family does not hold true, then we have a point $q \in Q$ such that moving it towards $\fopt(Q)$ would not change the output location of the mechanism. Then by Lemma~\ref{lem:towards} we can move $q$ towards $\fopt(Q)$ to create an instance $Q'$ with $r(Q', o(Q'), c) > r(Q, o(Q), c) \geq r(P, o(P), c) $, and we are done again. 
We now verify the five subproperties of property (4) for the \coa \ family. If the points are not clustered anymore on the axes in $Q$ after the paired movements, we can apply Lemma~\ref{lem:clusteringpair} to find an instance $Q'$ such that $r(Q', o(Q'), c) > r(Q, o(Q), c) \geq r(P, o(P), c) $. 
Therefore, we can assume that the paired movements move points from an existing cluster to another cluster, so the first, the fourth and the fifth subconditions clearly hold true. 
The second and the third subconditions can also be verified easily because we only reduce the number of points with negative $x$-coordinate, and we do not change the number of points with negative $y$-coordinate. 
Therefore, we conclude that by performing the paired movements, we can find an instance $Q$ either with $r(Q, o(Q), c) > r(P, o(Q), c)$ or we have $Q \in \coaC(n,c)$ with $r(Q, o(Q), c) \geq r(P, o(Q), c)$. 
We can then iteratively remove points from the optimal location and will finally reach an instance $Q \in \coaC(n,c)$ such that there are no points on the optimal location, i.e. we find an instance $Q \in \coaC(n,c)$, $r(Q, o(Q), c) \geq r(P, o(P), c)$ and $q \in A_x \cup A_{+y}$ for all $q \in Q$. 
Note that we have enough points on the $-x$ half axis to eliminate all points on the optimal location because we have argued before that $| \{ p \in P: p \in  A_{-x} \}| \geq |\{p \in P : p = o(P)\}|$, and this concludes the proof of the lemma.
\end{proof}

\lemgeotoaxis*
\begin{proof}
Let $P \in \coaC(n, c)$ be such that $p \in A_x \cup A_{+y}$ for all $p\in P$, so the points are all on $(-x, 0)$, $(0,x'')$, or $(x',0)$ for some $x, x',x'' > 0$. Without loss of generality, by rescaling, assume that $x'' = 1$.
Let $L, R, U \subseteq P$ be the points on  $(-x, 0)$, $(0,1)$, and $(x',0)$ respectively. Since $\fa(P, \fopt(P), c) = (0,0)$, we have  $|L| = cn$, $|U| = |R| = \frac{1-c}{2}n$. Let $Q$ be the instance that is the same as $P$ for all $i \notin U$, has $q_i = o(P)$ for all $i \in U$.

First note that by Lemma~\ref{lem:optunchanged}  we have $o(Q) = o(P)$. By Lemma~\ref{lem:lowerbound} we know that for any $c \in [0,1)$ there exists an instance $Q'$ such that $\approxratio(Q', \fopt(Q'), c) \geq \frac{\sqrt{2c^2 + 2}}{1 + c}$. If $\approxratio(P, \fopt(P), c) < \frac{\sqrt{2c^2 + 2}}{1 + c}$, then we have $\approxratio(Q', \fopt(Q'), c) > \approxratio(P, \fopt(P), c)$. Now assume $\approxratio(P, \fopt(P), c) \geq \frac{\sqrt{2c^2 + 2}}{1 + c}$. let $\Delta_f$ denote the cost decrease of the algorithm, i.e., $\Delta_f = C^u(\fa(P, \fopt(P), c), P)) - C^u(\fa(Q, \fopt(Q), c), Q)$, similarly, we let $\Delta_o = C^u(\fo(P), P)) - C^u(\fo(Q), Q).$ Showing that $\approxratio(Q, \fopt(Q), c) \geq \approxratio(P, \fopt(P), c)$ is equivalent to showing that $\frac{\sqrt{2c^2 + 2}}{1 + c} - \frac{\Delta_f}{\Delta_o} \geq 0$. To simplify the presentation of the lemma, we let  $(1-y_o)$ denote the points movement on $y$-axis and $\lambda(1-y_o)$ denote the points movement on $x$-axis. Then
\begin{align*}
    \Delta_f & = \frac{1-c}{2}n \cdot \lambda(1-y_o) - cn\lambda(1-y_o) + \frac{1-c}{2}n (1-y_o)\\
    \Delta_o &= \frac{1-c}{2}n\sqrt{(\lambda(1-y_o))^2 + (1-y_o)^2} = \frac{1-c}{2}n\sqrt{1+\lambda^2}(1-y_o)
    \end{align*}
    Let
$    g(c,\lambda) = \frac{\sqrt{2c^2 + 2}}{1 + c}  -\frac{\Delta_f}{\Delta_o} = \frac{\sqrt{2c^2+2}}{1+c} - \frac{(1-3c)\lambda+(1-c)}{(1-c)\sqrt{1+\lambda^2}}$.
Taking the first and second derivative w.r.t $\lambda$ we get
\[\frac{dg}{d\lambda} = - \frac{c\lambda-\lambda+1-3c}{\left(1-c\right)\left(\lambda^2+1\right)\sqrt{\lambda^2+1}} \quad \text{and} \quad \frac{d^2 g}{d\lambda^2} = -\frac{-2cx^2+2x^2+9cx-3x+c-1}{\left(-c+1\right)\left(x^2+1\right)^{\frac{5}{2}}}.\]
Solving $\frac{dg}{dx}= 0$ we get that $\lambda = \frac{1-3c}{1-c}$. 
If $c > \frac{1}{3}$, we have $\lambda < 0$ and since second derivative at $\lambda = \frac{1-3c}{1-c}$ is positive we get that for any $c > \frac{1}{3}$, $g(c, \lambda)$ is minimized at $\lambda = 0$. Which means that the movement on x-axis is $0$. In this case, we have $x_o(P) = 0$ and thus $P \in \ooaC(n, c)$. By Lemma~\ref{lem:ooa}, we then get that there exists $Q'$ such that either $ \approxratio(Q', \fopt(Q'), c) > \approxratio(P, \fopt(P), c)$ or $Q' \in \famPC(n, c) $ and  $ \approxratio(Q', \fopt(Q'), c) \geq \approxratio(P, \fopt(P), c)$.

Now if $c < \frac{1}{3}$,  we have $\lambda > 0$, and the second derivative of $g(c, \lambda)$ is positive, therefore given a specific $c $, $\lambda = \frac{1-3c}{1-c}$ is the minimizer. We then rewrite the function:
\[g(c) = \frac{\sqrt{2c^2 + 2}}{1 + c} - \frac{\sqrt{10c^2 - 8c +2}}{1-c}.\]
Again we take the derivative and set it to $0$,
\[\frac{dg}{dc} = \frac{2c-2}{\left(1+c\right)^2\sqrt{2+2c^2}}-\frac{6c-2}{\left(1-c\right)^2\sqrt{10c^2-8c+2}} = 0 \quad \Rightarrow c = 0, c = 0.301263\]
Again checking the second derivative we get that $c = 0$ is a minimizer and $c = 0.301263$ is a maximizer. Plug in $c = 0$ we get that
\[g(0) = \frac{\sqrt{2}}{1} - \frac{\sqrt{2}}{1} = 0 \quad \Rightarrow \quad \frac{\Delta_f}{\Delta_o} \leq \frac{\sqrt{2c^2 + 2}}{1 + c}.\]
and $\approxratio(Q, \fopt(Q), c) \geq \approxratio(P, \fopt(P), c)$.  By shifting all the points by $(-x_o(P), 0)$ to have $x_\fa(Q) = x_o(Q) = 0$, we obtain that  $Q \in \famPC(n, c)$.
\end{proof}

\lemytogm*

\begin{proof}
Let $P \in \coaR(n, c)$ be such that $p \in A_x \cup A_{+y}\cup \{o(P)\}$ for all $p\in P$, so the points are all on $(-x, 0)$, $(0,x'')$, $(x',0)$, or $o(P)$ for some $x, x',x'' > 0$. Without loss of generality, by rescaling, assume that $x'' = 1$.
Let $L, R, U,O \subseteq P$ be the points on  $(-x, 0)$, $(0,1)$, $(x',0)$, and $o(P)$ respectively. Let $k \in [0,1]$ be such that $|O| = kn$. Since $\fa(P, (0,0), c) = (0,0)$, we have  $|U|=|R| = (\frac{1+c}{2}-k)n$ and $|L| = (k-c)n$. Let $Q$ be the instance that is the same as $P$ for all $i \notin U$, and has $q_i = o(P)$ for all $i \in U$.

First note that by Lemma~\ref{lem:optunchanged} we have $o(Q) = o(P)$. By Lemma~\ref{lem:lowerbound} we know that for any $c \in [0,1)$ there exists an instance $Q'$ and prediction $\fpred$ such that $\approxratio(Q', \fpred, c) \geq \frac{\sqrt{2c^2 + 2}}{1 - c}$. If $\approxratio(P, (0,0), c) < \frac{\sqrt{2c^2 + 2}}{1 - c}$, then we have $\approxratio(Q', \fpred, c) > \approxratio(P, (0,0), c)$. Now assume $\approxratio(P, (0,0, c) \geq \frac{\sqrt{2c^2 + 2}}{1 - c}$. let $\Delta_f$ denote the cost decrease of the algorithm, i.e., $\Delta_f = C^u(\fa(P, (0,0),c), P) - C^u(\fa(Q, (0,0),c), Q)$, similarly, we let $\Delta_o = C^u(o(P), P) - C^u(o(Q), Q)$. The lemma statement is equivalent to $\frac{\sqrt{2c^2 + 2}}{1 - c} \geq \frac{\Delta_f}{\Delta_o}$. To simplify the presentation of the lemma, we let  $d$ denote the points movement on y-axis. Then
\begin{align*}
    \Delta_f &= (\frac{1+c}{2}-k)n \cdot d + (\frac{1+c}{2}-k)n \cdot x_o - (k-c)n \cdot x_o\\
    \Delta_o & = (\frac{1+c}{2}-k)n \sqrt{x_o^2 + d^2}
\end{align*}

Note that since Manhattan distance is a $\sqrt{2}$ approximation of the Euclidean distance, we have $r + d \leq \sqrt{2} \cdot \sqrt{x_o^2 + d^2}$. We get:

\begin{align*}
    \frac{\Delta_f}{\Delta_o} & = \frac{(\frac{1+c}{2}-k)n \cdot (d +r) - (k-c)n \cdot x_o}{(\frac{1+c}{2}-k)n \sqrt{x_o^2 + d^2}}\\
    &\leq  \frac{(\frac{1+c}{2}-k)n \cdot \sqrt{2} (\sqrt{x_o^2 + d^2}) - (k-c)n \cdot x_o}{(\frac{1+c}{2}-k)n \sqrt{x_o^2 + d^2}}\\
    &\leq \sqrt{2} - \frac{(k-c)n \cdot x_o}{(\frac{1+c}{2}-k)n \sqrt{x_o^2 + d^2}}\\
    & \leq \sqrt{2}
\end{align*}
Since $\approxratio(P, (0,0, c) \geq \frac{\sqrt{2c^2 + 2}}{1 - c} \geq \sqrt{2}$, we get that  $\approxratio(Q, \fopt(Q), c) \geq \approxratio(P, \fopt(P), c)$.
By shifting all the points by $(-x_o(P), 0)$ to have $x_\fa(Q) = x_o(Q) = 0$, we obtain that $Q \in \famPR(n, c)$.
\end{proof}

\lemcoaconsistency*
\begin{proof}
Let $c \in [0,1)$ and $n \in \N$. Consider an instance $P \in \coaC(n, c)$. By Lemma~\ref{lem:firstytox}, there exists $Q_1$ such that either $ \approxratio(Q_1, \fopt(Q_1), c)>  \approxratio(P, \fopt(P), c)$  or $Q_1 \in \coaC(n, c)$ and  $r(Q_1, \fopt(Q_1), c) \geq r(P, \fopt(P), c)$, and $q  \in A_x \cup A_{+y} \cup \{o(Q_1)\}$ for all $q \in Q_1$. By Lemma~\ref{lem:doublerotation}, there exists $Q_2$ such that either $ \approxratio(Q_2, \fopt(Q_2), c)>  \approxratio(Q_1, \fopt(Q_1), c)$  or $Q_2 \in \coaC(n, c) $ and $r(Q_2, \fopt(Q_2), c) \geq r(Q_1, \fopt(Q_1), c)$ and $q \in A_x \cup A_{+y}$ for all $q \in Q_2$. By Lemma~\ref{lem:geotoaxis}, there exists $Q_3$ such that either $ \approxratio(Q_3, \fopt(Q_3), c)>  \approxratio(Q_2, \fopt(Q_2), c)$  or $Q_3 \in  \famPC(n, c) $ and $r(Q_3, \fopt(Q_3), c) \geq r(Q_2, \fopt(Q_2), c)$. Thus, we have that  there either exists $n$ points $Q$ such that $ \approxratio(Q, \fopt(Q), c)>  \approxratio(P, \fopt(P), c)$ or that $\max_{Q \in \famPC(n, c) } \approxratio(Q, \fopt(Q), c) \geq \approxratio(P, \fopt(P), c)$.
\end{proof}

\lemcoarobustness*
\begin{proof}
Let $c \in [0,1)$ and $n \in \N$. Consider an instance $P \in \coaR(n, c)$. By Lemma~\ref{lem:firstytox}, there exists $Q_1$ such that either $ \approxratio(Q_1, (0,0), c)>  \approxratio(P, (0,0), c)$  or $Q_1 \in \coaR(n, c)$ and  $r(Q_1, (0,0), c) \geq r(P, \fopt(P), c)$, and $q  \in A_x \cup A_{+y} \cup \{o(Q_1)\}$ for all $q \in Q_1$. By Lemma~\ref{lem:ytogm}, there exists $Q_2$ and prediction $\fpred$ such that either $ \approxratio(Q_2, \fpred, c)>  \approxratio(Q_1, (0,0), c)$  or $Q_2 \in  \famPR(n, c) $ and $r(Q_2, (0,0), c) \geq r(Q_1, (0,0), c)$. Thus, we have  that  there either exists $n$ points $Q$ and prediction $\fpred$ such that $ \approxratio(Q, \fpred, c)>  \approxratio(P, (0,0), c)$ or that $\max_{Q \in \famPR(n, c) } \approxratio(Q, (0,0), c) \geq \approxratio(P, (0,0), c)$.
\end{proof}

\bibliographystyle{plainnat}
\bibliography{biblio}

\end{document}